\documentclass[a4paper,11pt]{article}

\usepackage[left=2.5cm,right=2.5cm,top=3cm,bottom=3cm]{geometry}
\usepackage[noadjust]{cite}
\usepackage{amsmath,amssymb,amsthm}
\usepackage{mathrsfs}
\usepackage{bm}
\usepackage{enumitem}
\usepackage{hyperref}

\usepackage{mathtools}

\newtheorem{theorem}{Theorem}
\numberwithin{theorem}{section}
\newtheorem{corollary}[theorem]{Corollary}
\newtheorem{lemma}[theorem]{Lemma}
\newtheorem{proposition}[theorem]{Proposition}

\theoremstyle{definition}
\newtheorem{definition}[theorem]{Definition}
\newtheorem{example}[theorem]{Example}
\newtheorem{remark}[theorem]{Remark}

\newcommand{\atp}{\otimes_\alpha}
\newcommand{\ptp}{\otimes_\pi}
\DeclareMathOperator{\ptpc}{\widehat{\otimes}_\pi}

\newcommand{\Bil}{\mathsf{Bil}\hspace{0.3mm}}
\newcommand{\Bm}{\hspace{0.2mm}\mathfrak{B}\hspace{0.2mm}}

\newcommand{\rk}{\mathrm{rk}\hspace{0.3mm}}

\newcommand{\fua}{\hspace{0.2mm}\mathfrak{X}\hspace{0.2mm}}
\newcommand{\fub}{\hspace{0.2mm}\mathfrak{Y}\hspace{0.2mm}}

\newcommand{\Qe}{\mathbb{Q}_{p,\mu}}
\newcommand{\Q}{\mathbb{Q}}
\newcommand{\R}{\mathbb{R}}
\newcommand{\N}{\mathbb{N}}
\newcommand{\Hi}{\mathcal{H}}
\newcommand{\Ki}{\mathcal{K}}
\newcommand{\JJ}{\mathcal{V}}
\newcommand{\Ba}[2]{\mathcal{B}(#1,#2)}
\newcommand{\Bad}[2]{\mathcal{B}_{\mathrm{ad}}(#1,#2)}
\newcommand{\M}{\mathsf{M}_\infty}
\newcommand{\CH}{\mathbb{H}}
\newcommand{\T}{\mathcal{T}}
\newcommand{\C}{\mathcal{C}}
\newcommand{\Cad}{\mathcal{C}_{\mathrm{ad}}}

\newcommand{\W}{\mathscr{W}}
\newcommand{\X}{\mathscr{X}}
\newcommand{\Y}{\mathscr{Y}}
\newcommand{\Ji}{\mathcal{L}}
\newcommand{\Ws}{\mathcal{W}}
\newcommand{\AU}{\overline{\mathcal{U}}(\Hi)}
\newcommand{\hi}{\mbox{\tiny $\Hi$}}
\newcommand{\ki}{\mbox{\tiny $\Ki$}}
\newcommand{\U}{\mathcal{U}(\Hi)}
\newcommand{\invset}{\mathfrak{S}}
\newcommand{\teniso}{\mathscr{S}}
\newcommand{\bmid}{\,\,\big\lvert\,\,}
\newcommand{\Bmid}{\,\,\Big\lvert\,\,}
\newcommand{\bbmid}{\,\,\bigg\lvert\,\,}
\newcommand{\scom}{`}
\newcommand{\dcom}{``}

\newcommand{\card}[1]{\mathrm{card}(#1)}
\newcommand{\bra}[1]{\langle#1\rvert} 
\newcommand{\ket}[1]{\lvert#1\rangle}
\newcommand{\Id}{\mathrm{Id}}
\newcommand{\inprd}{\langle\hspace{0.32mm}\cdot\hspace{0.6mm},\cdot\hspace{0.32mm}\rangle}
\newcommand{\HSprod}{\langle\hspace{0.32mm}\cdot\hspace{0.6mm},\cdot\hspace{0.32mm}\rangle_{\mbox{\tiny $\T(\Hi,\Ki)$}}}
\newcommand{\hsbraket}[2]{\langle#1,#2\rangle_{\mbox{\tiny $\T(\Hi,\Ki)$}}}
\newcommand{\braket}[2]{\langle#1,#2\rangle}
\newcommand{\func}[1]{\langle#1,\cdot\hspace{0.4mm}\rangle}
\newcommand{\No}[1]{\|#1\|}

\newcommand{\PNo}[1]{\|#1\|_\pi}
\newcommand{\abs}[1]{\left|#1\right|}
\newcommand{\smnor}{\mbox{\tiny $\|\hspace{-0.5mm}\cdot\hspace{-0.5mm}\|$}}
\newcommand{\nperp}{\overset{\smnor}{\perp}}
\newcommand{\ihs}{W}
\newcommand{\Noh}[1]{\No{#1}_{\hi}}
\newcommand{\Nok}[1]{\No{#1}_{\ki}}
\newcommand{\braketh}[2]{\braket{#1}{#2}_{\hi}}
\newcommand{\braketk}[2]{\braket{#1}{#2}_{\ki}}
\newcommand{\op}{\mathrm{op}}
\newcommand{\dom}{\mathrm{dom}}
\newcommand{\ran}{\mathrm{ran}}
\newcommand{\ketbra}[2]{|#1\rangle\langle#2|}
\newcommand{\hot}{\otimes}
\newcommand{\isoc}{\mathcal{P}}
\newcommand{\con}{\overline{(\,\cdot\,)}}

\newcommand{\J}[1]{\mathcal{J}_{\mbox{\tiny $#1$}}}
\newcommand{\tr}{\mathrm{tr}}
\newcommand{\B}{\mathscr{B}}
\DeclareMathOperator{\opE}{\mathnormal E}
\newcommand{\E}[3]{\hspace{-0.7mm}\sideset{^{#1}\hspace{-0.5mm}}{_{#3}^{#2}}\opE}
\DeclareMathOperator{\opT}{\mathnormal T}
\newcommand{\Top}[1]{\sideset{^{#1}}{}\opT}

\usepackage{xcolor}
\usepackage{soul}

\begin{document}

\title{The tensor product of $p$-adic Hilbert spaces}

\author{Paolo Aniello,$^{1,2}$\thanks{Email: paolo.aniello@na.infn.it} \hspace{1.2mm}
Lorenzo Guglielmi,$^{3,4}$\thanks{Email: lorenzo.guglielmi@unicam.it} \hspace{1.2mm}
Stefano Mancini,$^{3,4}$\thanks{Email: stefano.mancini@unicam.it} \hspace{0.6mm}
and Vincenzo Parisi\,$^{3,4}$\thanks{Email: vincenzo.parisi@unicam.it. (Present address: \textit{CONCEPT Lab, Fondazione Istituto Italiano di~Tecnologia,
via E.~Melen 83, Genova, 16152, Italy})}}

\date{
{\normalsize
$^1$\textit{Dipartimento di Fisica  ``Ettore Pancini'', Universit\`a  di Napoli ``Federico II'',\\
Complesso Universitario di Monte S.~Angelo, via Cintia, I-80126 Napoli, Italy}\\	
\vspace{3mm}
$^2$\textit{Istituto Nazionale di Fisica Nucleare, Sezione di Napoli,\\
Complesso Universitario di Monte S.~Angelo, via Cintia, I-80126 Napoli, Italy}\\
\vspace{3mm}
$^3$\textit{School of Science and Technology, University of Camerino,\\
via Madonna delle Carceri, 9, Camerino, I-62032, Italy}\\
\vspace{3mm}
$^4$ \textit{Istituto Nazionale di Fisica Nucleare, Sezione di Perugia,\\
via A.~Pascoli, I-06123 Perugia, Italy}
}}
\maketitle
	
\begin{abstract}
\noindent
In the framework of quantum mechanics over a quadratic extension of the ultrametric field of $p$-adic numbers, we introduce a notion of tensor product of $p$-adic Hilbert spaces. To this end, following a standard approach, we first consider the algebraic tensor product of $p$-adic Hilbert spaces. We next define a suitable norm on this linear space. It turns out that, in the $p$-adic framework, this norm is the analogue of the projective norm associated with the tensor product of real or complex normed spaces. Eventually, by metrically completing the resulting $p$-adic normed space, and equipping it with a suitable inner product, we obtain the tensor product of $p$-adic Hilbert spaces. That this is indeed the correct $p$-adic counterpart of the tensor product of complex Hilbert spaces is also certified by establishing a natural isomorphism between this $p$-adic Hilbert space and the corresponding Hilbert-Schmidt class. Since the notion of subspace of a $p$-adic Hilbert space is highly nontrivial, we finally study the tensor product of subspaces, stressing both the analogies and the significant differences with respect to the standard complex case. These findings should provide us with the mathematical foundations necessary to explore quantum entanglement in the $p$-adic setting, with potential applications in the emerging field of $p$-adic quantum information theory.
\end{abstract}

\tableofcontents

 	
\section{Introduction}

 During the final part of the 20th century, $p$-adic mathematical physics emerged as an effort to elaborate a non-Archimedean approach to space-time at very small scales. In fact, it is a general belief that standard quantum mechanics over the field of complex numbers is no more applicable below a scale of the order of Planck’s length ($l_\text{P}=\sqrt{\hbar G/c^3}\approx 10^{-35}\text{m}$), which is conjectured to be the smallest measurable length~\cite{peres1960quantum}. This belief led some theoretical physicists to consider the formulation of a non-Archimedean quantum physics as a viable description of quantum gravity and string theory~\cite{volovich1987p,volovich1987p2,aref1991beyond,khrennikov1991real}. The adoption of the field $\mathbb{Q}_p$ of $p$-adic numbers, where $p$ is a prime number, is strongly favored because it is, up to isomorphisms, the only complete non-Archimedean field obtainable by metrically extending rational numbers~\cite{gouvea1997p,robert2013course}.
 
 According to the models proposed so far, the field $\mathbb{Q}_p$ can be invoked at two different levels when formulating a $p$-adic quantum mechanical theory. On the one hand, one may simply resort to wave functions of the form $\psi\colon\mathbb{Q}_p\to \mathbb{C}$ (or, more generally, complex wave functions on a $p$-adic manifold), in such a way as to incorporate the non-Archimedean behavior of the theory directly into the space-time coordinates~\cite{ruelle1989quantum,meurice1989quantum, khrennikov1991p,albeverio1996p,vourdas2013quantum,vourdas2017finite,zuniga2024p1,zuniga2024p2}. On the other hand, $\mathbb{Q}_p$ --- or rather, for technical reasons related, e.g., to the role played by a suitable conjugation map in the theory, a \emph{quadratic extension} of this field  --- may be adopted as the base field of the carrier vector space of physical states~\cite{khrennikov1993statistical,albeverio1997representation,albeverio1997spectrum,albeverio1997fourier,albeverio1998regularization,albeverio2009p}. This is both a more subtle and more significant modification of the standard formulation of quantum mechanics that involves, among other aspects, the development of new mathematical tools and a suitable interpretation of the physically observable quantities~\cite{aniello2023trace,aniello2024quantum,aniello2023states}.

In a preceding paper~\cite{aniello2023trace}, a new approach to quantum mechanics over the non-Archimedean field of $p$-adic numbers has been proposed that mimics, to some extent, the algebraic formulation of standard quantum mechanics~\cite{strocchi2008introduction}. The main novelty of this approach lies in a suitable definition of a $p$-adic Hilbert space and its remarkable consequences.
A central aspect of our definition of a $p$-adic Hilbert space is that the existence of orthonormal bases is part of the fundamental axioms rather than an effect of these axioms. This peculiar fact is due to the very special interplay between the non-Archimedean norm and the scalar product~\cite{aniello2023trace}, which has no parallel in the standard complex case.

As mentioned previously, in our formulation of $p$-adic quantum mechanics, physical states are defined as linear functionals on a $*$-algebra of observables over a quadratic extension $\mathbb{Q}_{p,\mu}$ of the field of $p$-adic numbers, where $\mu$ is a parameter labeling the various possible quadratic extensions of $\mathbb{Q}_p$ (for generalities regarding the quadratic extensions of $p$-adic numbers, see~\cite{gouvea1997p,robert2013course}). At this stage, as pointed out in~\cite{aniello2023trace}, an important step forward in developing a complete theory of $p$-adic quantum mechanics is to define the \emph{tensor product} of two $p$-adic Hilbert spaces $\Hi$ and $\Ki$, to provide a mathematical description of a \emph{composite physical system}. In our present contribution, we address this issue precisely.

Let us briefly outline our strategy. First, following the usual approach, we define the \emph{algebraic} tensor product~\cite{ryan2002introduction} $\Hi\atp\Ki$ of two $p$-adic Hilbert spaces, which amounts to considering them as mere linear spaces over the field $\mathbb{Q}_{p,\mu}$. Next, a non-Archimedean norm is introduced, suitable for achieving a $p$-adic normed space $\Hi\ptp\Ki$ out of the initial algebraic tensor product $\Hi\atp\Ki$. This norm may be regarded as the ultrametric analogue of the standard \emph{projective norm}~\cite{ryan2002introduction} associated with the algebraic tensor product of two real or complex Banach spaces  --- which explains the meaning of the symbol $\Hi\ptp\Ki$ we adopt here --- and, accordingly, the associated Banach space completion will be denoted by $\Hi\ptpc\Ki$. Eventually, by endowing the $p$-adic Banach space $\Hi\ptpc\Ki$ with a suitable inner product, a $p$-adic Hilbert space is obtained; i.e., the ($p$-adic Hilbert space) tensor product $\Hi\otimes\Ki$ of the original Hilbert spaces $\Hi$ and $\Ki$. In particular, an orthonormal basis is explicitly constructed in such a way as to satisfy a fundamental axiom of a $p$-adic Hilbert space (i.e., the existence of such a basis). We also show that the tensor product $p$-adic Hilbert space $\Hi\otimes\Ki$ is isomorphic to the space of trace class operators --- which in turn, endowed with a suitable scalar product, may be identified with the space of Hilbert-Schmidt operators (the so-called $p$-adic trace/Hilbert-Schmidt class) --- from $\Hi$ into $\Ki$. Obtaining this result may be regarded as a sort of \emph{litmus test} that our construction (which, differently form the standard complex setting, involves the projective tensor product) provides indeed the correct $p$-adic analogue of the usual tensor product of two complex Hilbert spaces.

Finally, we study how the tensor product $\Hi\otimes\Ki$ behaves w.r.t.\ the subspaces of the original $p$-adic Hilbert spaces $\Hi$ and $\Ki$. In this regard, among many other technical issues, a fundamental point that should be considered is the following. The notion of subspace itself is not as simple and straightforward in the $p$-adic setting as it is in the complex case~\cite{aniello2024quantum}. Accordingly, the analogies --- as well as several nontrivial differences --- with respect to the standard complex setting, are analyzed throughout the article.

The paper is structured as follows. In Section~\ref{sec.2}, we recall the main definitions and properties of $p$-adic Hilbert spaces, as well as some basic facts about linear operators between $p$-adic Hilbert spaces that are relevant for our purposes. Section~\ref{sec.3} is the core part of the paper, since it is devoted to the construction of the tensor product of two $p$-adic Hilbert spaces, along the lines traced above. Next, in Section~\ref{sec.4}, we prove the aforementioned isomorphism between the tensor product $\Hi\otimes\Ki$ and the $p$-adic trace/Hilbert-Schmidt class of operators from $\Hi$ into $\Ki$. The implementation of this isomorphism involves a notion of anti-unitary operator suitable for the $p$-adic framework. In Section~\ref{sec.5}, we analyze the properties of the tensor product of two $p$-adic Hilbert spaces w.r.t.\ the subspaces of the two factor spaces. Finally, in Section~\ref{sec.6}, conclusions are drawn,
and we take a look at some research prospects.


\section{Preliminaries}
\label{sec.2}

In this section, we collect some basic results and tools that will be relevant for our later purposes. To begin with, we discuss the $p$-adic normed and Hilbert spaces over a \emph{quadratic extension} $\Qe$ of the field $\Q_p$ of $p$-adic numbers (for a comprehensive discussion of $\Q_p$ and its quadratic extensions, the reader may refer, e.g., to~\cite{gouvea1997p,robert2013course}). We next focus on the classes of bounded, adjointable, and trace class/Hilbert-Schmidt operators between $p$-adic Hilbert spaces.

Even if the reader may find a more detailed account in the preceding papers~\cite{aniello2023trace,aniello2024quantum,aniello2023states} --- where, for the sake of simplicity, linear operators defined on a \emph{single} Hilbert space are considered --- 
it is worth discussing here how the main results can be extended to operators acting, in general, between \emph{two different} $p$-adic Hilbert spaces --- say, $\Hi$ and $\Ki$ --- as this generalization will provide us with an essential tool for studying, in Section~\ref{sec.4}, the relation between the $p$-adic trace/Hilbert-Schmidt class operators from $\Hi$ to $\Ki$ and the tensor product $\Hi\otimes\Ki$ of the two Hilbert spaces.

\subsection{\texorpdfstring{$p$}{Lg}-adic Hilbert spaces}\label{subsec.2.1}

We recall that a \emph{$p$-adic normed space} is a pair $(X,\No{\cdot})$, where $X$ is a vector space over $\Qe$, and $\No{\cdot}$ an \emph{ultrametric} (or non-Archimedean) norm defined on $X$, namely, a map $\No{\cdot}\colon X\rightarrow\R^+$ that is positive definite, absolute-homogeneous and satisfies the so-called \emph{strong triangle inequality}, i.e., 
\begin{equation}\label{1}
\No{x+y}\leq \max(\No{x}\,,\No{y}),\quad \forall x,y\in X.
\end{equation}
(More generally, we say that $\No{\cdot}\colon X\rightarrow\R^+$ is an \textit{ultrametric seminorm} if it is absolute-homogeneous and satisfies relation~\eqref{1}.)
If $X$ is complete w.r.t.\ the \emph{(ultra)-metric} induced by the norm $\No{\cdot}$, the pair $(X,\No{\cdot})$ is called a \emph{$p$-adic Banach space}~\cite{robert2013course,aniello2023trace,van1978non,aniello2024quantum,perez2010locally,narici71,schikhof2007ultrametric}.
In the following, we will consider \emph{separable} $p$-adic Banach spaces only.

\begin{definition}\label{def.orth}
Let $(X,\No{\cdot})$ be a $p$-adic normed space. We say that two vectors $x,y$ in $X$ are (mutually) \emph{norm-orthogonal} --- and we write $x\nperp y$ --- if the condition 
 \begin{equation}\label{eq.2n}
    \No{\alpha x+\beta y}=\max(\No{\alpha x}\,,\No{\beta y}) 
 \end{equation}
 holds for every $\alpha,\beta\in\Qe$. We say that $\{x,y\}\subset X\setminus\{0\}$ is a \emph{norm-orthogonal system} if $x$ and $y$ are mutually norm-orthogonal; in particular, we say that $\{x,y\}$ is a \emph{normal system} if $x$ and $y$ are mutually norm-orthogonal, and, in addition, they are \emph{normalized}, i.e., $\No{x}=\No{y}=1$.
\end{definition}

\begin{remark}
The previous notion of orthogonality of vectors is equivalent to the so-called \emph{Birkhoff-James orthogonality}~\cite{Birkhoff,James}; i.e., $x\nperp y$ iff $\No{x}\le\No{x+\alpha y}$ for all $\alpha\in\Qe$ (clearly, in a real or complex Hilbert space this notion is equivalent to the usual one). The geometric intuition behind the latter definition is the following: The vector $x$ is orthogonal to the vector $y$ iff every triangle with one side equal to $x$ and another side parallel to $y$ has its third side that is not shorter than $x$.
Note that we have: $0\nperp x$, for all $x\in X$; $x\nperp x\;\implies\;x=0$; if $x\nperp y$, then $\alpha x\nperp \beta y$ for all $\alpha,\beta\in\Qe$.
\end{remark}

More generally, given a $p$-adic normed space $(X,\No{\cdot})$, we say that a (nonempty) set $\mathfrak{A}\subset X\setminus\{0\}$ is a \emph{norm-orthogonal system} if either $\mathfrak{A}$ is a singleton or, for any $x,y\in\mathfrak{A}$, with $x\neq y$, $\{x,y\}$ is a norm-orthogonal system; we say that $\mathfrak{A}$ is a \emph{normal system} if either $\mathfrak{A}=\{x\}$, with $\No{x}=1$, or, for any $x,y\in\mathfrak{A}$, with $x\neq y$, $\{x,y\}$ is a normal system.
Otherwise stated, a (nonempty) set $\mathfrak{A}\subset X\setminus\{0\}$ is a normal system if it consists of normalized vectors that are pairwise norm-orthogonal.

\begin{definition}\label{def.basis}
A countable subset $\mathfrak{B}$ of a (separable) $p$-adic Banach space $(X,\No{\cdot})$ is called a \emph{norm-orthogonal basis} (alternatively, a \emph{normal basis}) if $\mathfrak{B}$ is a norm-orthogonal (respectively, a normal) system and, moreover, if for each $x\in X$ there exists a map $\alpha_x\colon \mathfrak{B}\rightarrow \Qe$ such that
$x$ can be expressed --- in a unique way, up to reordering of the summands --- as the sum of an unconditionally norm-convergent series of the form
\begin{equation}
x=\sum_{b\in\mathfrak{B}}\alpha_x(b)\hspace{0.6mm} b.
\end{equation}
\end{definition}

\begin{remark}
Note that a norm-orthogonal basis in a $p$-adic Banach space is an \emph{unconditional} Schauder basis~\cite{aniello2023trace}.    
\end{remark}

If $\mathfrak{B}$ is a norm-orthogonal basis in a $p$-adic Banach space $(X,\No{\cdot})$, then,
for every vector $x\in X$, we have~\cite{aniello2023trace}:
\begin{equation} \label{norvec}
\|x\|=\sup_{b\in\mathfrak{B}}|\alpha_x(b)|\hspace{0.6mm} \|b\|
=\max_{b\in\mathfrak{B}}|\alpha_x(b)|\hspace{0.6mm} \|b\|.
\end{equation}

One can prove that a (separable) $p$-adic Banach space $(X,\No{\cdot})$ always admits a norm-orthogonal basis, while a normal basis exists in $X$ iff the additional condition $\No{X}=\abs{\Qe}$ is satisfied~\cite{robert2013course,aniello2023trace,perez2010locally} --- where $\No{X}$ and $\abs{\Qe}$ are subsets of $\R^+$ defined by
\begin{equation}
\No{X}\coloneqq\{\No{x}\mid x\in X\},\quad\abs{\Qe}\coloneqq\{\abs{\alpha}\mid\alpha\in\Qe\}.
\end{equation}
In the following, we will call a (separable) $p$-adic Banach space $(X,\No{\cdot}$) admitting a normal basis (equivalently, such that $\No{X}=\abs{\Qe}$) a \emph{normal $p$-adic Banach space}.
\begin{remark}\label{rem.2.1n}
A less stringent notion of orthogonality in a $p$-adic normed space $(X,\No{\cdot})$ is given by the so-called \emph{$t$-orthogonality}~\cite{perez2010locally}. Specifically, we say that two vectors $x,y$ in $X$ are \emph{$t$-orthogonal} --- and we write $x\nperp_t y$ --- if the condition
\begin{equation}\label{eq.5n}
\No{\alpha x+\beta y}\geq t\,\max\big\{\No{\alpha x},\No{\beta y}\big\}, \quad t\in (0,1],
\end{equation}
holds for every $\alpha,\beta\in\Qe$. Note that $x\nperp_t y$ iff $y\nperp_t x$. In the case where $t=1$, relation~\eqref{eq.5n} entails the orthogonality condition~\eqref{eq.2n}; that is, $1$-orthogonality coincides with the  --- previously introduced --- notion of orthogonality in Definition~\ref{def.orth}. We say that $\{x,y\}\subset X\setminus \{0\}$ is a $t$-orthogonal system if $x$ and $y$ are (mutually) $t$-orthogonal vectors, while a (nonempty) subset $\mathfrak{N}$ of $X\setminus\{0\}$ is $t$-orthogonal, if either $\mathfrak{N}$ is a singleton or, for any $x,y\in\mathfrak{N}$, with $x\neq y$, $\{x,y\}$ is a $t$-orthogonal system. Finally, in the same vein of  Definition~\ref{def.basis}, a system $\bm{h}\equiv\{h_i\}_{i\in I}$ in $X$ is called a $t$-orthogonal basis (of X), if $\bm{h}$ is a $t$-orthogonal system and, additionally, if any element $x$ in $X$ can be expressed, in a unique way, as the sum of a norm-convergent series of the form $x=\sum_{i\in I}\alpha_ih_i$, where $\alpha_i\in\Qe$.
\end{remark}
\begin{remark}\label{rem.2.1}
Let $(X,\No{\cdot})$ be a $p$-adic Banach space. We recall that the set $c_0(I,X)$ --- where $I$ is a countable index set, i.e., $I=\{1,\ldots,n\}$ or $I=\N$ --- of \emph{zero-convergent} sequences in $X$ is defined as~\cite{robert2013course,aniello2023trace,aniello2024quantum}
\begin{equation}
    c_0(I,X)\coloneqq\{\{x_i\}_{i\in I}\mid x_i\in X,\;\mbox{$\lim_i x_i=0$}\}
\end{equation}
--- in the case where $I$ is finite, we set $\lim_i x_i\equiv 0$. It is possible to prove that $c_0(I,X)$, endowed with the so-called \emph{sup-norm} $\No{\cdot}_\infty$, defined for any $x=\{x_i\}_{i\in I}\in c_0(I,X)$ as
\begin{equation}
    \No{x}_\infty\coloneqq \sup_{i\in I}\No{x_i},
\end{equation}
is a $p$-adic Banach space~\cite{robert2013course,aniello2023trace,aniello2024quantum}. In particular, with $X=\Qe$, viewed as a one dimensional vector space equipped with the norm $\abs{\;\cdot\;}$, we have the Banach space 
\begin{equation}
c_0(I,\Qe)\coloneqq\{\{\alpha_i\}_{i\in I}\mid \alpha_i\in \Qe,\;\mbox{$\lim_i \alpha_i=0$}\}.
\end{equation}
\end{remark}
We next recall the notion of a $p$-adic Hilbert space~\cite{aniello2023trace,aniello2024quantum}. As in the standard complex setting, this notion stems from a suitable definition of an inner product, \emph{but}, remarkably, in the $p$-adic setting the norm and the inner product are not related in the usual way.

Let $(X,\No{\cdot})$ be a $p$-adic Banach space (or, more generally, a $p$-adic normed space) over a quadratic extension $\Qe$ of $\Q_p$. We call a map $\inprd\colon X\times X \rightarrow \Qe$ a \emph{non-Archimedean inner product} on $X$, if it is \emph{sesquilinear}, \emph{Hermitian}, and satisfies the \emph{Cauchy-Schwarz inequality}, i.e., 
\begin{equation}
\abs{\braket{x}{y}}\leq \No{x}\No{y}, \quad \forall x,y\in X.
\end{equation} 
 We will call a triple $(X,\No{\cdot},\inprd)$ --- where $(X,\No{\cdot})$ is a $p$-adic Banach space, and $\inprd$ a non-Archimedean inner product --- an \emph{inner-product $p$-adic Banach space}. 
\begin{definition}
Let $(X,\No{\cdot},\inprd)$ be an inner-product $p$-adic Banach space. We say that two vectors $x,y$ in $X$ are \emph{inner-product-orthogonal} (in short, IP-orthogonal), and we write $x\perp y$, 
if $\braket{x}{y}=0$. 
\end{definition}
\begin{remark}
It is worth observing that in an inner-product $p$-adic Banach space, in general,  $\No{x}\neq\sqrt{\abs{\braket{x}{x}}}$, i.e., the ultrametric norm does not stem directly from the non-Archimedean inner product~\cite{aniello2023trace,aniello2024quantum}. Accordingly, in the $p$-adic setting, we have two distinct natural notions of orthogonality: the IP-orthogonality and the (previously introduced) norm-orthogonality. 
 \end{remark}
 A countable set of vectors $\Psi\equiv\{\psi_i\}_{i\in I}$ in an inner-product $p$-adic Banach space $(X,\No{\cdot})$ is said to be an \emph{orthonormal basis}, if $\Psi$ is a normal basis in $X$ (see Definition~\ref{def.basis}) and if, moreover, $\braket{\psi_i}{\psi_j}=\delta_{ij}$, for every pair of indices $i,j\in I$. 
In an inner product $p$-adic Banach space, the existence of an orthonormal basis is \emph{not} guaranteed. Nevertheless, when such a basis exists, several salient properties of complex Hilbert spaces have analogs in the p-adic setting.
It is precisely this observation that led us to set the following~\cite{aniello2023trace,aniello2024quantum}:
\begin{definition}\label{hilbertspace}
We call a quadruple $(\Hi, \No{\cdot},\inprd,\Psi)$, where $(\Hi, \No{\cdot},\inprd)$ is an inner-product $p$-adic Banach space, and $\Psi\equiv\{\psi_i\}_{i\in I}$ an orthonormal basis in $\Hi$, a \emph{$p$-adic Hilbert space}. We put $\dim(\Hi)\coloneqq\card{I}$.
\end{definition}
 \begin{remark}
 Let $(\Hi,\No{\cdot},\inprd,\Psi)$ be a $p$-adic Hilbert space ($\dim(\Hi)=\card{I}$). Owing to the continuity of the inner product w.r.t.\ both its arguments, it is not difficult to verify that any $x\in\Hi$
 admits a (unique) representation --- w.r.t.\ the orthonormal basis $\Psi$ --- as the sum on an unconditional norm convergent series of the form
 \begin{equation}\label{eq.7}
    x=\sum_{i\in I} \braket{\psi_i}{x}\psi_i.
\end{equation}
In particular, from relation~\eqref{eq.7}, one can derive the \emph{non-Archimedean Parseval} identity~\cite{aniello2023trace}:
\begin{equation}\label{eq.11n}
\No{x}=\max_{i\in I}\abs{\braket{\psi_i}{x}}.
\end{equation}
\end{remark}

\begin{example}\label{exa.2.4}
Consider the $p$-adic Banach space $(c_0(I,\Qe), \No{\cdot}_\infty)$ (see Remark~\ref{rem.2.1}). This space admits a normal basis $\bm{e}=\{e_i\}_{i\in I}$ consisting of all the sequences $e_i$ in $c_0(I,\Qe)$ of the form
\begin{equation}
    e_1=(1,0,0,\ldots),\quad e_2=(0,1,0,\ldots),\quad e_3=(0,0,1,\ldots),\quad\ldots\quad.
\end{equation}
 We can endow $c_0(I,\Qe)$ with the so-called \emph{canonical inner product} associated with $\bm{e}\equiv\{e_i\}_{i\in I}$, i.e., the map defined, for every pair of vectors $x=\sum_{i\in I}x_ie_i$, and $y=\sum_{j\in I}y_je_j$ in $c_0(I,\Qe)$, as
\begin{equation}\label{eq.10}
c_0(I,\Qe)\times c_0(I,\Qe)\ni (x,y)\mapsto\braket{x}{y}_c\coloneqq\sum_{i\in I}\overline{x}_iy_i\in \Qe,
\end{equation}
where the overline denotes the conjugate in $\Qe$.
By construction, we have: $\braket{e_i}{e_j}_c=\delta_{ij}$, $\forall i,j\in I$, i.e., $\bm{e}\equiv\{e_i\}_{i\in I}$ is an orthonormal basis. Therefore, $(c_0(I,\Qe),\No{\cdot}_\infty,\inprd_c,\bm{e})$ is a $p$-adic Hilbert space; in the literature~\cite{aniello2023trace,aniello2024quantum,diag2016,khren1990} it is sometimes referred to as the \emph{coordinate $p$-adic Hilbert space}, and denoted by $\mathbb{H}(I)$. 
\end{example}

\begin{remark}
If $\Hi$ is a $p$-adic Hilbert space, then $\No{\Hi}=\abs{\Qe}$, i.e., $\Hi$ is a normal $p$-adic Banach space. Conversely, given a normal $p$-adic Banach space $X$, and a normal basis $\bm{e}\equiv\{e_i\}_{i\in I}$, we can always endow $X$ with a suitable inner product --- the \emph{canonical inner product} associated with $\{e_i\}_{i\in I}$, see Example~\ref{exa.2.4} --- such that $\bm{e}$ becomes an orthonormal basis, and $X$ a $p$-adic Hilbert space~\cite{aniello2023trace,aniello2024quantum}.
\end{remark}

Let $(\Hi,\No{\cdot}_{\hi},\inprd_{\hi},\Phi)$ and $(\Ki,\No{\cdot}_{\ki},\inprd_{\ki},\Psi)$ be two $p$-adic Hilbert spaces over $\Qe$. We call a linear map $\ihs\colon \Hi\rightarrow \Ki$ an \emph{isomorphism} of $p$-adic Hilbert spaces if it is surjective, and satisfies the following conditions:
\begin{enumerate}[label=\tt{(A\arabic*)}]

\item $\No{\ihs x}_{\ki}=\No{x}_{\hi},\quad \forall x\in\Hi$;

\item $\braketk{\ihs x}{\ihs y}=\braketh{x}{y},\quad \forall x,y\in \Hi$;

\end{enumerate}
i.e., $W$ is an inner-product-preserving surjective isometry~\cite{aniello2023trace,aniello2024quantum}. In complete analogy with the standard complex setting, it is possible to prove that two $p$-adic Hilbert spaces $\Hi$ and $\Ki$ are isomorphic iff $\dim(\Hi)=\dim(\Ki)$~\cite{aniello2023trace,aniello2024quantum}. 

\begin{example}\label{exa.2.13}
Consider the coordinate $p$-adic Hilbert space $\mathbb{H}(I)$ introduced in Example~\ref{exa.2.4}. If $\Hi$ is a $p$-adic Hilbert space, and $\Psi=\{\psi_i\}_{i\in I}$ an orthonormal basis, it is not difficult to see that the map defined as
\begin{equation}\label{eq.12}
W_\Psi\colon \Hi\ni x=\sum_{i\in I}\braket{\psi_i}{x}\psi_i\mapsto W_\Psi(x)\coloneqq\{\braket{\psi_i}{x}\}_{i\in I}\equiv\breve{x}\in\mathbb{H}(I) 
\end{equation}
provides an isomorphism (the so-called \emph{canonical isomorphism}) of $\Hi$ onto $\mathbb{H}(I)$~\cite{aniello2023trace,aniello2024quantum}.
\end{example}
\begin{remark}\label{rem.2.4}
 Recall that the \emph{topological dual} $\Hi^\prime$ of a Hilbert space $\Hi$ is the set of all bounded linear functionals on $\Hi$. If $\Hi$ is a $p$-adic Hilbert space such that $\dim(\Hi)=\infty$, $\Hi^\prime$ is not a $p$-adic Hilbert space, but a $p$-adic Banach space isomorphic to the set $\ell^\infty(I,\Qe)$ of bounded sequences in $\Qe$~\cite{aniello2023trace,aniello2024quantum},
 \begin{equation}
     \ell^\infty(I,\Qe)\coloneqq\big\{\eta\equiv\{\eta_i\}_{i\in I}\mid \eta_i\in \Qe,\, |\eta_i|<\infty\big\}.
 \end{equation}
In fact, the mapping defined by
\begin{equation}
\J{\Hi}\colon \Hi\ni\psi\mapsto\J{\Hi}\psi\coloneqq\func{\psi}\in\Hi^\prime,
\end{equation}
provides a \emph{conjugate-linear isometry} of $\Hi$ into its dual $\Hi^\prime$, that is surjective iff $\dim(\Hi)<\infty$. The isomorphism between $\Hi^\prime$ and the ultrametric Banach space $\ell^\infty(I,\Qe)$ --- with $\card{I}=\dim(\Hi^\prime)$ --- is realized via the surjective isometry defined, for any given orthonormal basis $\Psi\equiv\{\psi_i\}_{i\in I}$ in $\Hi$, by
\begin{equation}
\mathcal{E}_{\Psi}\colon \ell^\infty(I,\Qe)\ni\xi=\{\xi_i\}_{i\in I}\mapsto\sum_{i\in I} \xi_i\func{\psi_i}=\sum_{i\in I}\xi_i\J{\Hi}\psi_i\in\Hi^\prime,
\end{equation}
where, if $I=\N$, the series converges in the weak$^\ast$-topology (see Proposition~$3.36$ in~\cite{aniello2023trace}).
\end{remark}

\subsection{Linear operators between \texorpdfstring{$p$}{Lg}-adic Hilbert spaces}\label{subsec.2.2}

We now turn our attention to some classes of bounded operators from a $p$-adic Hilbert space $\Hi$ into another $p$-adic Hilbert space $\Ki$. The case where
$\Hi=\Ki$ has already been analyzed in~\cite{aniello2023trace}, and the general case we are considering here can be treated by means of a straightforward modification of the methods developed therein.

Hereafter, for the sake of definiteness, it is worth assuming that $\dim(\Hi)=\dim(\Ki)=\infty$, and, accordingly, we will identify the index set $I$ of the previous subsection with $\N$. In all other cases --- $\dim(\Hi)<\dim(\Ki)=\infty$, etc.\ --- analogous results hold with rather obvious modifications. For the sake of notational simplicity, we will adopt the shorthand notations $c_0$, $\ell^\infty$ and $\CH$ for $c_0(\N,\Qe)$, $\ell^\infty(\N,\Qe)$ and $\CH(\N)$, respectively.

\subsubsection{Bounded and adjointable operators}\label{subsubsec.2.2.1}
As in the standard complex setting, linear operators between p-adic Hilbert spaces are also characterized by their domain --- a linear space that may not coincide with the original Hilbert space, though it is a linear subspace thereof --- which is typically specified prior to defining the operator itself.
Let $\Hi$ and $\Ki$ be two $p$-adic Hilbert spaces, and let $B\colon\dom(B)\rightarrow \Ki$, with $\dom(B)\subset\Hi$, be a linear operator. It is often useful to use the shorthand notation $B\colon\Hi\dashrightarrow\Ki$ to highlight the \emph{initial} space $\Hi$ and the \emph{final} space $\Ki$, rather than the \emph{domain} $\dom(B)\subset\Hi$ and \emph{range} $\ran(B)\coloneqq B(\dom(B))\subset\Ki$.

We say that $B\colon\Hi\dashrightarrow\Ki$ is a \emph{bounded} operator if
\begin{equation}\label{eq.15}
\No{B}\coloneqq\sup_{\dom(B)\ni x\neq 0}\frac{\No{Bx}}{\No{x}}<\infty.
\end{equation}
As in the complex setting, $B\colon\Hi\dashrightarrow\Ki$ is bounded iff it is \emph{continuous}~\cite{aniello2023trace,aniello2024quantum}; moreover, if $B$ is bounded, it can be extended by continuity to the norm-closure of its domain $\dom(B)$, which is a $p$-adic Banach space, but, in general, \emph{not} a $p$-adic Hilbert space. Therefore, in the $p$-adic setting (bounded) operators whose domain coincides with the full initial space deserve special attention, and we adopt an \emph{ad hoc} terminology and notation.

Consequently, we will say that a linear operator $B\colon\Hi\dashrightarrow\Ki$ is \textit{all-over}, if $\dom(B)=\Hi$, and, in such a case, we will write $B\colon\Hi\rightarrow\Ki$; otherwise stated, by writing $B\colon\Hi\rightarrow\Ki$ we mean precisely that $\dom(B)=\Hi$.

We will denote by $\Ba{\Hi}{\Ki}$ the set of all-over, bounded --- equivalently, continuous --- linear operators from $\Hi$ to $\Ki$; i.e., $\Ba{\Hi}{\Ki}\coloneqq\{B\colon\Hi\rightarrow\Ki\mid \mbox{$B$ bounded}\}$.

It turns out that, in the $p$-adic setting, the linear operators in $\Ba{\Hi}{\Ki}$ are best described and characterized as \emph{matrix operators} w.r.t.\ any given pair of orthonormal bases in the initial space $\Hi$ and in the final space $\Ki$. We then briefly introduce this (larger) class of linear operators.

Let $(A_{mn})$, $m,n\in\N$, be an infinite matrix in $\Qe$ (we will denote by $\M(\Qe)$ the set of all such matrices), and let $\Phi\equiv\{\phi_n\}_{n\in \N}$ and $\Psi\equiv\{\psi_m\}_{m\in \N}$ be orthonormal bases in $\Hi$ and $\Ki$, respectively. The matrix $(A_{mn})$ and the orthonormal bases $\Phi$, $\Psi$, determine
a matrix operator $\op_{\Phi,\Psi}(A_{mn})\colon\Hi\dashrightarrow\Ki$ defined by
\begin{align}\label{eq.13}
\dom(\op_{\Phi,\Psi}(A_{mn}))\coloneqq 
&
\Big\{\mbox{$\zeta=\sum_{n}z_n\phi_n\in\Hi\Bmid \sum_n A_{mn}z_n$ converges, $\forall m\in \N$, and}
\nonumber\\
&
\mbox{$\sum_{m}\big(\sum_n A_{mn}z_n\big)\psi_m$ converges too, i.e., $\big\{\sum_n A_{mn}z_n\big\}_{m\in \N}\in c_0$}\Big\},
\end{align}
\begin{equation}\label{eq.18s}
\op_{\Phi,\Psi}(A_{mn})\zeta\coloneqq \mbox{$\sum_m\big(\sum_n A_{mn}\braket{\phi_n}{\zeta}\big)\psi_m$},\quad \forall\zeta\in\dom(\op_{\Phi,\Psi}(A_{mn})).
\end{equation}

\begin{remark}
We stress that the matrix operator $\op_{\Phi,\Psi}(A_{mn})\colon\Hi\dashrightarrow\Ki$ is --- in general --- \emph{unbounded}.
\end{remark}

Let us focus on the case of a matrix operator $B=\op_{\Phi,\Psi}(B_{mn})\colon\Hi\rightarrow\Ki$, i.e., $\dom(B)=\Hi$. It is natural to wonder whether such an operator is bounded and, in the case where the answer is affirmative, whether \emph{every} bounded operator $B\in\Ba{\Hi}{\Ki}$ is a matrix operator of this form.

In analogy with the notation introduced in~\cite{aniello2023trace} --- where matrix operators acting on a single Hilbert space are considered --- we denote by $(\Hi,\Ki)_{\Phi,\Psi}$ the set of the \emph{all-over matrix operators} from $\Hi$ to $\Ki$, associated with the pair of orthonormal bases $\Phi\subset\Hi$ and $\Psi\subset\Ki$; i.e., we put 
\begin{equation}
(\Hi,\Ki)_{\Phi,\Psi}\coloneqq\big\{\op_{\Phi,\Psi}(B_{mn})\colon\Hi\dashrightarrow\Ki\mid (B_{mn})\in\M(\Qe),\,\,\dom(\op_{\Phi,\Psi}(B_{mn}))=\Hi\big\}.
\end{equation}
Actually, the set $(\Hi,\Ki)_{\Phi,\Psi}$ does not depend on a particular choice of the pair of orthonormal bases $\Phi\equiv\{\phi_n\}_{n\in\N}\subset\Hi$ and $\Psi\equiv\{\psi_m\}_{m\in\N}\subset\Ki$. Moreover, it can be proven~\cite{aniello2023trace} that the matrix operator $\op_{\Phi,\Psi}(B_{mn})$ is all-over iff 
\begin{equation}\label{eq.bondop}
    \lim_m B_{mn}=0,\;\forall n\in\N,\;\;\mbox{and}\;\;\sup_{m,n}\abs{B_{mn}}<\infty.
\end{equation}
Consider now a bounded operator $B\in\Ba{\Hi}{\Ki}$. Clearly, we have:
\begin{equation}\label{eq.bondop2}
\lim_m \braket{\psi_m}{B\phi_n}=0, \;\forall n\in \N,\,\,\mbox{and}\,\,\sup_{m,n}\abs{\braket{\psi_m}{B\phi_n}}\leq \No{B}.
\end{equation}
Therefore, setting $B_{mn}\equiv\braket{\psi_m}{B\phi_n}$, the matrix operator $\op_{\Phi,\Psi}(B_{mn})$ is all-over, and, in fact, one can easily check that 
\begin{equation}\label{eq.matop}
B=\sum_m\sum_nB_{mn}\braket{\phi_n}{\cdot}\psi_m=\sum_{n}\sum_{m}B_{mn}\func{\phi_n}\psi_m=\op_{\Phi,\Psi}(B_{mn}).
\end{equation}
Here, $(\braket{\phi_n}{\cdot}\psi_m)\zeta\coloneqq\braket{\phi_n}{\zeta}\psi_m$ --- i.e., $\braket{\phi_n}{\cdot}\psi_m\equiv \ketbra{\psi_m}{\phi_n}$ in Dirac's notation --- and the iterated series converge w.r.t.\ the strong operator topology of $\Ba{\Hi}{\Ki}$. Hence, we have that $\Ba{\Hi}{\Ki}\subset (\Hi,\Ki)_{\Phi,\Psi}$.

The crucial observation, at this point, is that the reverse inclusion holds too; i.e., one can prove the following result, which is a natural generalization of Theorem~$4.2$ in~\cite{aniello2023trace}:
\begin{theorem}\label{th.2.5}
Let $\Hi$ and $\Ki$ be $p$-adic Hilbert spaces. For every pair of orthonormal bases $\Phi=\{\phi_n\}_{n\in \N}\subset\Hi$ and $\Psi\equiv\{\psi_m\}_{m\in \N}\subset\Ki$, the following relation holds:
\begin{align}
\Ba{\Hi}{\Ki}
& =
(\Hi,\Ki)_{\Phi,\Psi}
\nonumber \\  \label{eq.20}
& =\big\{\op_{\Phi,\Psi}(B_{mn})\colon\Hi\dashrightarrow \Ki\mid \mbox{\rm $\lim_m B_{mn}=0$,  $n\in\mathbb{N}$, and  $\sup_{m,n}|B_{mn}|<\infty$}\big\}.
\end{align}
Moreover, every bounded operator $B\in\Ba{\Hi}{\Ki}$ can be expressed as in~\eqref{eq.matop}, and we have:
\begin{equation}\label{eq.21}
\No{B}=\sup_{m,n}\abs{B_{mn}}=\sup_n\No{B\phi_n}.
\end{equation}
\end{theorem}
\begin{proof}
See Appendix~\ref{app.a1}.
\end{proof}

By Theorem~\ref{th.2.5}, $\Ba{\Hi}{\Ki}$ coincides with the set $(\Hi,\Ki)_{\Phi,\Psi}$ of the all-over matrix operators from $\Hi$ into $\Ki$ --- independently of the choice of the orthonormal bases $\Phi\equiv\{\phi_n\}_{n\in\N}$ and $\Psi\equiv\{\psi_m\}_{m\in\N}$ --- and, moreover, the infinite matrices in $\M(\Qe)$ corresponding to bounded operators in $\Ba{\Hi}{\Ki}$ admit a complete characterization.

In the standard complex setting, every bounded operator admits a (Hilbert space) \emph{adjoint}, which is also a bounded operator. In the $p$-adic setting --- since the \emph{dual} $\Hi^\prime$ of $\Hi$ cannot be identified with $\Hi$ itself --- the notion of operator adjoint is not as straightforward as in the complex case. However, there is a distinguished class of bounded operators --- called \emph{adjointable} (bounded) operators --- that may be regarded as the true
$p$-adic counterpart of the bounded linear operators between complex Hilbert spaces.

Let us first recall that a bounded operator $B\in\Ba{\Hi}{\Ki}$ will admit a (generalized) \emph{adjoint} $B^\prime\in\Ba{\Ki^\prime}{\Hi^\prime}$ defined by
\begin{equation}
\Ki^\prime\ni \omega\mapsto B^\prime \omega \coloneqq\omega\circ B \in\Hi^\prime.
\end{equation}
However, it is possible to single out, within $\Ba{\Hi}{\Ki}$, a suitable class of bounded operators that admit a \emph{proper} Hilbert space adjoint; i.e., a suitable bounded operator in $\Ba{\Ki}{\Hi}$. In fact, with every $B\in\Ba{\Hi}{\Ki}$ we can associate a linear operator $B^\dagger\colon\Ki\dashrightarrow\Hi$  --- the so-called \emph{pseudo-adjoint} of $B$ --- defined as follows. 
We first put
\begin{equation}\label{eq.27}
\dom\big(B^\dagger\big)\coloneqq\{\psi\in\Ki\mid\mbox{
$\exists\, \zeta\in\Hi$
s.t.\ $\braket{\psi}{B\phi}=\braket{\zeta}{\phi}$, 
$\forall \phi\in\Hi=\dom(B)$} \}.
\end{equation}
It is clear that $\dom(B^{\dagger})$ is a linear subspace of $\Ki$; it can be regarded --- via the conjugate-linear isometry $\J{\Hi}$  (see Remark~\ref{rem.2.4}) --- as the set of all vectors $\psi\in\Ki$ such that the bounded functional $\braket{\psi}{B(\cdot)}\in\Hi^\prime$  can be identified with an element $\eta$ of $\Hi$, i.e., 
\begin{equation}
\psi\in\dom(B^\dagger)\;\;\iff\;\;\braket{\psi}{
B(\cdot)}=\func{\eta}=\J{\Hi}\eta.
\end{equation}
The condition 
\begin{equation}\label{eq.29}
  \big(\J{\Ki}\psi\big)(B\phi)=\braket{\psi}{B\phi}=\braket{\eta}{\phi}=\big(\J{\Hi}\eta\big)(\phi), \quad \forall\psi\in\dom(B^\dagger),\,\forall\phi\in\Hi
\end{equation}
 uniquely determines a linear operator $B^\dagger\colon\Ki\dashrightarrow\Hi$, whose domain $\dom\big(B^\dagger\big)$ is defined by~\eqref{eq.27}, by setting $B^\dagger\psi\coloneqq \eta$ (in fact, since $\J{\Hi}$ is a conjugate linear isometry, there exists a \emph{unique} vector $\eta\in\Hi$ such that $\J{\Hi}\eta=\big(\J{\Ki}\psi\big)\circ B$). Moreover, by~\eqref{eq.29}, one has that  
 \begin{equation}\label{eq.30}
\No{B^\dagger\psi}=\No{\eta}=\No{\J{\Hi}\eta}=\No{(\J{\Ki}\psi)\circ B}\leq\No{\J{\Ki}\psi}\,\No{B}=\No{\psi}\,\No{B},
\end{equation}
i.e., $\No{B^\dagger}\leq\No{B}$ and we see that $B^\dagger$ is a bounded operator.

One can show (along the lines traced in~\cite{aniello2023trace}) that either $\dom\big(B^\dagger\big)=\Ki$ or $\dom\big(B^\dagger\big)$ is a proper closed subspace of $\Ki$. Taking into account this fact, we now set the following: 
\begin{definition}
Let $B\in\Ba{\Hi}{\Ki}$. We say that $B$ is \emph{adjointable} if $\dom(B^\dagger)=\Ki$, i.e., if the bounded operator $B^\dagger$ --- the pseudo-adjoint of $B$ --- is all-over. In such a case, we write $B^\ast$ instead of $B^\dagger$ and call $B^\ast$ the \emph{adjoint} of $B$ in $\Ba{\Hi}{\Ki}$.    
\end{definition}
In what follows, we will denote by $\Bad{\Hi}{\Ki}$ the set of all adjointable bounded operators from $\Hi$ to $\Ki$. We now aim to provide a suitable matrix characterization of an adjointable operator. In particular, we state the following:
\begin{theorem}\label{th.2.7}
Let $B\in\Bad{\Hi}{\Ki}$, and let $B^\ast\colon\Ki\rightarrow \Hi$ be its adjoint. Then, $B^\ast\in\Bad{\Ki}{\Hi}$; i.e., $B^\ast$ is both bounded and adjointable. Specifically, given any pair of orthonormal bases $\Phi\equiv\{\phi_n\}_{n\in\N}$ in $\Hi$ and $\Psi\equiv\{\psi_m\}_{m\in\N}$ in $\Ki$, if $B=\op_{\Phi,\Psi}(B_{mn})\in\Bad{\Hi}{\Ki}$, then
\begin{equation}
B^\ast=\op_{\Psi,\Phi}(B^\ast_{mn}),
\end{equation}
where we have put $B_{mn}^\ast=\overline{B_{nm}}$, and the coefficients of the matrix $(B_{mn})$ associated with $B$ satisfy the following conditions:
\begin{enumerate}[label=\tt(\arabic*)]

\item \label{cond.ad.i}
$\sup_{m,n}\abs{B_{mn}}<\infty$;

\item \label{cond.ad.ii}
$\lim_m B_{mn}=0$, $\forall n\in\N$;

\item \label{cond.ad.iii}
$\lim_n B_{mn}=0$, $\forall m\in\N$.

\end{enumerate}
Moreover, it is clear that
\begin{equation}\label{norm.adj}
    \No{B^\ast}=\sup_{m,n}\abs{B_{mn}}=\No{B}.
\end{equation}

Conversely, if $B\colon\Hi\dashrightarrow\Ki$ is a matrix operator --- i.e., if $B=\op_{\Phi,\Psi}(B_{mn})$, for some pair of orthonormal bases $\Phi\equiv\{\phi_n\}_{n\in\N}\subset\Hi$ and $\Psi\equiv\{\psi_m\}_{m\in\N}\subset\Ki$ --- and, moreover, the coefficients of the matrix $(B_{mn})\in\M(\Qe)$  satisfy conditions \ref{cond.ad.i}--\ref{cond.ad.iii} above, then we have that $B\colon\Hi\rightarrow\Ki\in\Bad{\Hi}{\Ki}$, and, therefore, the adjoint $B^\ast\colon\Ki\rightarrow\Hi\in\Bad{\Ki}{\Hi}$ of $B$ can be expressed as a matrix operator of the form $B^\ast=\op_{\Psi,\Phi}(B^\ast_{mn})=\op_{\Psi,\Phi}(\overline{B_{nm}})$.
\end{theorem}
\begin{proof}
See Appendix~\ref{app.a1}.
\end{proof}
\begin{remark}
From Theorem~\ref{th.2.7} it follows that, if $B=\op_{\Phi,\Psi}(B_{mn})$ 
is an adjointable operator, then $B^\ast$ is adjointable too, and $(B^\ast)^\ast=B$.
Moreover, $\Bad{\Hi}{\Ki}$ is a linear subspace of $\Ba{\Hi}{\Ki}$, as one can easily check by observing that for all $A,B\in\Bad{\Hi}{\Ki}$, and all $\alpha\in\Qe$, $\alpha A$ and $A+B$ are in $\Bad{\Hi}{\Ki}$.
\end{remark}

\subsubsection{The \texorpdfstring{$p$}{}-adic trace/Hilbert-Schmidt class}
\label{subsubsec.2.2.2}

As is well known, in the standard complex setting, the trace class and the Hilbert-Schmidt class, are two distinct operator spaces: a Banach space, in the first case, and a Hilbert space (endowed with the Hilbert-Schmidt scalar product) in the second~\cite{barry2015operator}, the former space being strictly contained in the latter in the case where the carrier Hilbert space is infinite-dimensional. In the $p$-adic setting, instead, we have a \emph{unique} operator space that plays the role of both the trace and the Hilbert-Schmidt class. However, we will first introduce this space as a $p$-adic Banach space --- the so-called \emph{$p$-adic trace class} $\T(\Hi,\Ki)$ --- and we will next argue that this space, endowed with a suitable scalar product, becomes a $p$-adic Hilbert space; i.e., the \emph{$p$-adic Hilbert-Schmidt class}. Another distinguishing feature of the $p$-adic trace class $\T(\Hi,\Ki)$ is that it is a Banach space w.r.t.\ the standard operator norm (rather than some `trace norm').

Let $\Hi$ and $\Ki$ be $p$-adic Hilbert spaces, with orthonormal bases $\Phi\equiv\{\phi_m\}_{m\in\N}$, and $\Psi\equiv\{\psi_n\}_{n\in\N}$, respectively. Extending, in a natural way, the definition~\cite{aniello2023trace,aniello2024quantum} of the trace class $\T(\Hi)$ of a single $p$-adic Hilbert space, we now introduce the set of \emph{trace class operators} from $\Hi$ to $\Ki$ as follows:
\begin{equation}\label{eq.37}
\T(\Hi,\Ki)\coloneqq\mbox{$\{\op_{\Phi,\Psi}(T_{mn})\colon\Hi\dashrightarrow\Ki\mid T_{mn}\in\M(\Qe)\,\,\text{s.t.}\,\lim_{m+n}T_{mn}=0\}$}.
\end{equation}

At this point, a further comment is in order. The reader will be familiar with the standard definition of a trace class operator on a \emph{complex} Hilbert space, which crucially relies on the notion of \emph{positive} operator and, accordingly, of the \emph{absolute value} of a bounded operator (see, e.g., Chapter~3 of~\cite{barry2015operator}). In the $p$-adic setting there is no meaningful notion of positivity, and the previous definition of a trace class operator turns out to be the most convenient one.

\begin{remark}
Note that the defining condition for a trace class operator in the $p$-adic setting --- i.e., $T=\op_{\Phi,\Psi}(T_{mn})\in\T(\Hi,\Ki)$ iff $\lim_{m+n}T_{mn}=0$ --- involves a pair of orthonormal bases $\Phi$, $\Psi$, because $T$ is defined as a matrix operator, but the property of being a trace class operator is actually \emph{independent} of the choice of this pair of bases; see below. Also note that, if $\lim_{m+n}T_{mn}=0$, then the sum $\sum_m T_{mm}$ of all `diagonal entries' is a convergent series in $\Qe$.
\end{remark}

Exploiting the defining condition~\eqref{eq.37}, and by suitably adapting to the present setting the line of reasoning in Sec.~VI of~\cite{aniello2023trace} (see, in particular, Propositions~$6.2$ and~$6.7$, and Corollary~$6.8$), the forthcoming result can be proved. We will denote by $\JJ$ any \emph{third} Hilbert space over $\Qe$.
\begin{theorem}
The definition~\eqref{eq.37} of the set $\T(\Hi,\Ki)$ does not depend on the choice of the pair of orthonormal bases $\Phi$, $\Psi$. Moreover, $\T(\Hi,\Ki)$ is a Banach subspace of $\Bad{\Hi}{\Ki}$, and 
\begin{align}
T\in\T(\Hi,\Ki)\subset\Bad{\Hi}{\Ki} &
\implies T^\ast\in\T({\Ki},{\Hi}),
\nonumber \\  \label{fundrels}
T\in\T(\Hi,\Ki), \ A\in\Bad{\Ki}{\JJ}, \ B\in\Bad{\JJ}{\Hi} &
\implies AT\in\T({\Hi},{\JJ}),\ TB\in\T({\JJ},{\Ki}).
\end{align}
Finally, every trace class operator $T=\op_{\Phi,\Psi}(T_{mn})\in\T(\Hi,\Ki)$ is \emph{traceable}; that is, for any pair of orthonormal bases $\Phi\equiv\{\phi_m\}_{m\in\N}$ in $\Hi$ and $\Psi\equiv\{\psi_n\}_{n\in\N}$ in $\Ki$, the \emph{trace} $\tr(T)$ of $T$ --- defined as the sum of the convergent series
\begin{equation}\label{eq.trac}
\tr(T)\coloneqq\sum_m\braket{\psi_m}{T\phi_m}=\sum_m T_{mm}\in\Qe,
\end{equation}
--- does not depend on the choice of the orthonormal bases $\Phi\equiv\{\phi_m\}_{m\in\N}$ and $\Psi\equiv\{\psi_n\}_{n\in\N}$.
\end{theorem}
\begin{remark}
Let $\tr(\hspace{0.2mm}\cdot\hspace{0.2mm})\colon\T(\Hi,\Ki)\rightarrow\Qe$ be the map that associates, with every operator $T\in\T(\Hi,\Ki)$, its trace $\tr(T)$ as defined in~\eqref{eq.trac}. One can easily check that $\tr(\hspace{0.2mm}\cdot\hspace{0.2mm})$ is a linear functional on $\T(\Hi,\Ki)$; namely, for any $S,T\in\T(\Hi,\Ki)$ and $\alpha\in\Qe$,
\begin{equation}
\tr(S+T)=\tr(S)+\tr(T),\quad \tr(\alpha T)=\alpha\tr(T).
\end{equation}
Moreover, for every $T\colon\Hi\rightarrow\Ki\in\T(\Hi,\Ki)$, the adjoint $T^\ast\in\T(\Ki,\Hi)$ satisfies the usual relation $\tr(T^\ast)=\overline{\tr(T)}$.
\end{remark}

Exactly as in the standard complex setting, a linear operator $C\colon\Hi\rightarrow \Ki$ is said to be \emph{compact} if it maps the unit ball $\Hi_1\coloneqq\{\psi\in\Hi\mid\No{\psi}\leq 1\}$ of $\Hi$ onto a \emph{precompact} subset $C(\Hi_1)$ of $\Ki$; that is, if $\overline{C(\Hi_1)}^{\No{\cdot}}$ is a compact subset of $\Ki$~\cite{aniello2023trace,van1978non}. In the following, we will denote by $\C(\Hi,\Ki)\subset\Ba{\Hi}{\Ki}$ the set of all compact (bounded) operators from $\Hi$ to $\Ki$. 

A special class of compact operators is formed by the so-called \emph{block-finite operators}~\cite{aniello2023trace}. Precisely, a matrix operator $B=\op_{\Phi,\Psi}(B_{mn})\colon\Hi\dashrightarrow\Ki$ is said to be \emph{block-finite} w.r.t.\ a certain pair of orthonormal bases $\Phi\subset\Hi$, $\Psi\subset\Ki$ --- and we denote the set of all such operators by $\B_{\Phi,\Psi}$ --- if there exists some $k\in\N$ such that 
\begin{equation}
\max\{m,n\}>k\;\;\implies\;\;B_{mn}=0.
\end{equation}
A block-finite operator has finite rank; therefore, it is adjointable and of trace class. 
\begin{theorem}
The following characterization of the trace class $\T(\Hi,\Ki)$ holds true:
\begin{equation}
\T(\Hi,\Ki)=\Cad(\Hi,\Ki)\coloneqq\C(\Hi,\Ki)\cap\Bad{\Hi}{\Ki}.
\end{equation}
Moreover, for every pair of orthonormal bases $\Phi\equiv\{\phi_m\}_{m\in\N}$ in $\Hi$ and $\Psi\equiv\{\psi_n\}_{n\in\N}$ in $\Ki$, $\T(\Hi,\Ki)$ coincides with the uniform closure of the set $\B_{\Phi,\Psi}$ of block-finite operators w.r.t.\ the orthonormal bases $\Phi$, $\Psi$; i.e., 
\begin{equation}
\T(\Hi,\Ki)=\overline{\B_{\Phi,\Psi}}^{\No{\cdot}}.
\end{equation}
\end{theorem}
\begin{proof}
The first claim is proved likewise Theorem~$6.35$ in~\cite{aniello2023trace}, whereas the second is a straightforward adaptation of
Theorem~$6.34$  in~\cite{aniello2023trace} to this specific case.
\end{proof}
As anticipated, in the concluding part of this subsection, we argue that $\T(\Hi,\Ki)$, endowed with a suitable sesquilinear form, has a natural structure of a $p$-adic Hilbert space, namely, the \emph{$p$-adic Hilbert-Schmidt space}. In fact, let us introduce the sesquilinear form
\begin{equation}\label{eq.46}
\T(\Hi,\Ki)\times \T(\Hi,\Ki)\ni (S,T)\rightarrow \hsbraket{S}{T}\coloneqq\tr(S^\ast T)\in\Qe.
\end{equation}
Here, note that, by relations~\eqref{fundrels}, if $S,T\in\T(\Hi,\Ki)$, then $S^\ast T\colon\Hi\rightarrow\Hi\in\T(\Hi)\equiv\T(\Hi,\Hi)$, so that the preceding mapping is well defined. It is clear that $\HSprod$ is Hermitian, as follows by observing that 
\begin{equation}
\hsbraket{S}{T}=\tr(S^\ast T)=\overline{\tr(T^\ast S)}=\overline{\braket{T}{S}}_{\mbox{\tiny $\T(\Hi,\Ki)$}}.
\end{equation}
Following the terminology adopted in~\cite{aniello2023trace}, we call the Hermitian sesquilinear form $\HSprod$ the \emph{Hilbert-Schmidt product} of $\T(\Hi,\Ki)$. 
    
The final step is the explicit construction of an orthonormal basis for the inner product $p$-adic Banach space $\big(\T(\Hi,\Ki),\HSprod\big)$. To this end, given any pair of orthonormal bases $\Phi\equiv\{\phi_m\}_{m\in\N}\subset\Hi$ and $\Psi\equiv\{\psi_n\}_{n\in\N}\subset\Ki$, we introduce the family of matrix operators $\big\{\E{jk}{\Phi,\Psi}{}\colon\Hi\rightarrow\Ki\big\}_{j,k\in\N}$ defined by
\begin{equation}\label{eq.48}
\E{jk}{\Phi,\Psi}{}\coloneqq\op_{\Phi,\Psi}\big(\E{jk}{}{mn}\big),\quad\mbox{where $\E{jk}{}{mn}=\delta_{jm}\delta_{kn}$},
\end{equation}
 i.e., $\E{jk}{\Phi,\Psi}{}=\func{\phi_k}\psi_j$, or equivalently $\E{jk}{\Phi,\Psi}{}=\ketbra{\psi_j}{\phi_k}$ in Dirac's notation.
 Note that, for every $j,k\in\N$, $\E{jk}{\Phi,\Psi}{}$, is a rank-one operator and, therefore, it is in the trace class $\T(\Hi,\Ki)$. 

Eventually, we are ready to formulate the following result, which is a natural generalization of Theorem~$6.38$ in~\cite{aniello2023trace}, that holds in the single Hilbert space setting. 
\begin{theorem}\label{th.tracla}
The quadruple
\begin{equation}
\Big(\T(\Hi,\Ki), \No{\cdot},\HSprod,\big\{\E{jk}{\Phi,\Psi}{}\big\}_{j,k\in\N}\Big)
\end{equation}
--- where $\HSprod$ is the Hilbert-Schmidt product defined in~\eqref{eq.46}, and $\big\{\E{jk}{\Phi,\Psi}{}\big\}_{j,k\in\N}$ the set of trace class operators introduced in~\eqref{eq.48}, that form an orthonormal basis in $\T(\Hi,\Ki)$ --- is a $p$-adic Hilbert space. 
\end{theorem}

\begin{proof}
See Appendix~\ref{app.a1}.
\end{proof}

    
\section{Construction of the tensor product}
\label{sec.3}

In this section, we proceed to the construction of the tensor product of $p$-adic Hilbert spaces. As in the real or complex setting, we first define the \emph{algebraic} tensor product of two vector spaces over $\Qe$. We stress that, even if we deal with $p$-adic Hilbert spaces, in this preliminary stage we do not need more than the mere linear properties of these spaces. Next, to endow the algebraic tensor product with a metric structure, we introduce an appropriate ultrametric norm. This is a nontrivial step, because what is the appropriate norm is not clear in advance, and, in fact, it turns out that the analogy with the tensor product of complex Hilbert spaces is not helpful in this case. Finally, after taking the completion w.r.t.\ the metric induced by this norm, we bring up a suitable inner product and an orthonormal basis, in such a way as to obtain eventually a \emph{new} $p$-adic Hilbert space; i.e., the tensor product of the initial $p$-adic Hilbert spaces. 

\begin{definition}
Let $X,Y,Z$ be vector spaces over the field $\Qe$. Recall that a map $\Bm$ from the Cartesian product $X\times Y$ into $Z$ is said to be \emph{bilinear} if it is linear w.r.t.\ each of its arguments; that is
\begin{align}
\Bm(\alpha_1 x_1+\alpha_2 x_2,y)
& =
\alpha_1 \Bm(x_1,y)+\alpha_2 \Bm(x_2,y), 
\nonumber \\
\Bm(x,\beta_1 y_1+\beta_2 y_2)
& =
\beta_1 \Bm(x,y_1)+\beta_2 \Bm(x,y_2),
\end{align}
for all vectors $x,x_1,x_2\in X$, $y,y_1,y_2\in Y$, and for all scalars 
$\alpha_i,\beta_i \in \Qe$, $i=1,2$. We will denote by $\Bil(X\times Y, Z)$ the vector space of all bilinear maps from $X\times Y$ to $Z.$ In the special case where $Z=\Qe$ (the one-dimensional $\Qe$-linear space), we put $\Bil(X\times Y)\equiv\Bil(X\times Y, \Qe)$ (the space of bilinear forms on $X\times Y$). 
\end{definition}

The algebraic tensor product $X\atp Y$ of the vector spaces $X$ and $Y$ can be defined as a space of linear functionals on $\Bil(X\times  Y)$, in the usual fashion~\cite{ryan2002introduction}. 

\begin{definition}
Let $X,Y$ be vector spaces over $\Qe$, and let $x\in X$, $y\in Y$. The \emph{simple tensor} $x\otimes y\colon\Bil( X\times Y)\rightarrow\Qe$ is the linear functional defined by
\begin{equation}
(x\otimes y)(\Bm)\coloneqq\Bm(x,y),
\end{equation}
for every bilinear form $\Bm\in\Bil(X\times Y)$. 

The \emph{algebraic tensor product} $X\atp Y$ is defined as the linear subspace of the \emph{algebraic dual} space $\Bil(X\times Y)^\vee$ spanned by all simple tensors $x\otimes y$. Therefore, a generic element $u\in X \atp Y$ can be expressed (not uniquely) as a finite linear combination of the form
\begin{equation} \label{lincomb}
u=\sum_{i=1}^n \lambda_i\, x_i \otimes y_i
\end{equation}
--- with $n\in \mathbb{N}$, $\lambda_i\in \Qe, x_i\in X$ and $y_i\in Y$, $i\in\{1,\ldots, n\}$ --- so that 
\begin{equation}
u(\Bm)=\sum_{i=1}^n \lambda_i\, \Bm(x_i, y_i)\in\Qe.
\end{equation}
\end{definition} 

Next result is the so-called \emph{universal property of the algebraic tensor product}~\cite{ryan2002introduction}.

\begin{theorem}
For every bilinear map $\Bm:X\times Y\to \Qe$, there exists a unique linear map $B\colon X \atp Y\to \Qe$ such that $\Bm(x,y)=B(x \otimes y)$, for all $x\in X$ and $y\in Y$. 
\end{theorem}
\begin{proof}
Let $\Bm:X\times Y\to \Qe$ be a bilinear map and define 
$B:X\atp Y\to \Qe$ by 
\begin{equation}
B\Big(\sum_{i=1}^n x_i \otimes y_i\Big)=\sum_{i=1}^n \Bm(x_i,y_i).
\end{equation}
Suppose now that $\sum_{i=1}^n x_i \otimes y_i=0.$ Then, for each linear functional $\varphi$ on $\Qe$, the composition $\varphi \circ B$ is a bilinear functional on $X\times Y$ and 
\begin{equation}
\varphi\Big(\sum_{i=1}^n \Bm(x_i,y_i)\Big)=\sum_{i=1}^n\varphi\circ\Bm(x_i,y_i)=0.
\end{equation}
Therefore, $B$ is well defined and the bilinear map $\Bm$ associated with the linear map $B$ yields an identification of $\Bil(X\times Y)$ and linear maps from $X\atp Y$ to $\Qe$.
\end{proof}

We now list some useful properties of the tensor product. 
\begin{proposition} \label{propro}
The mapping $X \times Y\ni(x,y)\mapsto x\otimes y\in X \atp Y$ satisfies the following properties:
\begin{enumerate}[label=\rm(\roman*)]

\item  \label{i}
$(x_1+x_2)\otimes y=x_1 \otimes y+x_2 \otimes y$;

\item  \label{ii}
$x\otimes (y_1+y_2)=x\otimes y_1+x\otimes y_2$;

\item  \label{iii}
$\lambda(x\otimes y)=(\lambda x)\otimes y=x\otimes (\lambda y)$;

\item  \label{iv}
$0 \otimes y=x\otimes 0=0$;

\item \label{v} if $E\subset X$ and $F\subset Y$ are linearly independent sets, then $\{x \otimes y: x\in E, y\in F\}$ is a linearly independent subset of $X\atp Y$;

\item  \label{vi}
if $x_1,\dots ,x_n\in X$ and $y_1,\dots ,y_n\in Y$ --- with $y_1,\dots ,y_n$ linearly independent and $x_i\ne 0$, for all $i\in\{1,\ldots, n\}$ --- then $\{x_1\otimes y_1,\dots, x_n\otimes y_n\}$ is a linearly independent set.
 \end{enumerate}
\end{proposition}

\begin{proof}
Properties~\ref{i}--\ref{iv} hold by basic calculations. The proof of~\ref{v} and~\ref{vi} follow by a slight adaptation of the analogous results from Proposition~{1.1} in~\cite{ryan2002introduction} and Theorem~{10.1.6} in~\cite{perez2010locally}.
\end{proof}

By~\eqref{lincomb} and~\ref{iii}, every tensor $u\in X\atp Y$ can be written in the form 
\begin{equation}\label{eq110}
 u=\sum_{i=1}^n x_i \otimes y_i ,
\end{equation}
for some $n\in \N$. Moreover, in the case where $X$ and $Y$ are finite-dimensional vector spaces, by~\ref{v} we have that 
\begin{equation}
\dim (X \atp Y)=\dim(X)\dim(Y) .
\end{equation}

\begin{remark}
We stress that, even if decomposition~\eqref{eq110} is not unique, there is a smallest number $\rk(u)$ --- called the \textit{rank} of the tensor $u\in X\atp Y$ --- of simple tensors (rank-one tensors) that appear in this decomposition.
\end{remark}

In addition to the properties~\ref{i}--\ref{vi} listed in Proposition~\ref{propro}, we have the following further useful fact, namely, a characterization of the null vector in $X \atp Y$.

\begin{proposition} \label{proptens}
Let $X,Y$ be vector spaces over $\Qe$, and let $X^\vee,Y^\vee$ be the associated algebraic dual spaces. Then, the following statements are equivalent for $u=\sum_{i=1}^n x_i \otimes y_i \in X\atp Y$:
\begin{enumerate}[label=\rm(vii-\arabic*)]

\item \label{vii}
$u=0$;

\item \label{viii}
$\sum_{i=1}^n \fua (x_i)\, \fub (y_i)=0$, for every $\fua\in X^\vee$, $\fub\in Y^\vee$;

\item \label{ix}
$\sum_{i=1}^n \fua (x_i)\,y_i=0$, for every $\fua\in X^\vee$;

\item \label{x}
$\sum_{i=1}^n \fub(y_i)\,x_i=0$, for every $\fub\in Y^\vee$.

\end{enumerate}
\end{proposition}

\begin{proof}
The proof is again a natural adaptation of that in the complex setting, see e.g.\ Proposition 1.2 in~\cite{ryan2002introduction}.
\end{proof}

 Until now, we did not need any metric property of the $p$-adic vector spaces $X$ and $Y$, because the construction of the algebraic tensor product $X\atp Y$  is independent of any notion of norm. Since, however, we are interested in the tensor product of $p$-adic \emph{Hilbert} spaces, from now on we will assume that we actually deal with a pair of $p$-adic Hilbert spaces, for which we use the usual shorthand notation $\Hi$, $\Ki$ instead of the more cumbersome $(\Hi,\No{\cdot}_{\Hi},\inprd_{\Hi},\Phi)$, $(\Ki,\No{\cdot}_{\Ki},\inprd_{\Ki},\Psi)$ (see Definition~\ref{hilbertspace}).

Therefore, the next step of our construction is to endow the algebraic tensor product $\Hi\atp \Ki$ with a suitable \emph{ultrametric norm}, thereby obtaining, upon completion, a suitable $p$-adic Banach space. 

Recall that in the standard setting where the tensor product $\mathscr{H}\otimes\mathscr{K}$ of two \emph{complex} Hilbert spaces $\mathscr{H}$, $\mathscr{K}$ is considered, the norm of the Hilbert space $\mathscr{H}\otimes\mathscr{K}$ is precisely the one canonically associated with the scalar product of this space~\cite{Moretti}. 
 
In other words, a suitable scalar product is defined on the algebraic tensor product $\mathscr{H}\atp\mathscr{K}$, yielding a pre-Hilbert space; the Hilbert tensor product $\mathscr{H}\otimes\mathscr{K}$ is then obtained by completing this space.
Such an approach cannot be pursued in the $p$-adic setting, because, as noted previously, the norm and the scalar product are not related in the usual fashion --- but only by the Cauchy-Schwarz inequality. It is also worth noting that the canonical (Hilbert space) norm on $\mathscr{H}\otimes\mathscr{K}$ \emph{does not coincide} with the \emph{projective norm}~\cite{ryan2002introduction} on the algebraic tensor product $\mathscr{H}\atp\mathscr{K}$ (the linear spaces $\mathscr{H}$, $\mathscr{K}$ being endowed with their respective Hilbert space norms); namely, the (Hilbert space) norm of $\mathscr{H}\otimes\mathscr{K}$, when restricted to $\mathscr{H}\atp\mathscr{K}$, does not coincide with the following:
\begin{equation}
\PNo{u}\coloneqq\inf \bigg\{ \sum_{i=1}^n \No{x_i}_{\mathscr{H}}\, \No{y_i}_{\mathscr{K}}\bbmid u=\sum_{i=1}^n x_i\otimes y_i,\,x_i\in\mathscr{H},y_i\in\mathscr{K}\bigg\}.
\end{equation}
Here, the \emph{infimum} is taken over all possible representations of $u\in\mathscr{H}\atp\mathscr{K}$ (i.e., all possible decompositions of $u$ as a sum of elementary tensors).

A natural non-Archimedean counterpart of this norm, compatible with the strong triangle inequality characterizing the $p$-adic setting, is given by the following candidate norm:
\begin{equation}\label{RNo}
\PNo{u}\coloneqq \inf\bigg\{\max_{0\le i\le n}\Noh{x_i} \, \Nok{y_i}\bbmid u=\sum_{i=1}^n x_i\otimes y_i,\,x_i\in\Hi,y_i\in\Ki\bigg\}.
\end{equation}
Here, the \emph{infimum} is taken over all possible representations of $u\in\Hi\atp\Ki$. 

We now demonstrate that~\eqref{RNo} indeed defines a non-Archimedean norm. To this end, we first prove that the mapping $\Hi\atp\Ki\ni u \mapsto\PNo{u}\in\mathbb{R^+}$ in~\eqref{RNo} is an ultrametric seminorm on the algebraic tensor product $\Hi\atp\Ki.$

\begin{proposition}
Let $\Hi$ and $\Ki$ be two $p$-adic Hilbert spaces. Then, $\PNo{\cdot}$ is an ultrametric seminorm on the algebraic tensor product $\Hi\atp\Ki$. 
\end{proposition}

\begin{proof}
Let us start with absolute homogeneity. We need to prove that, for every scalar $\lambda$ and every vector $u=\sum_{i=1}^n x_i\otimes y_i\in\Hi\atp\Ki$, $\PNo{\lambda u}=\abs{\lambda}\PNo{u}$. For $\lambda=0$, this is trivial. If $\lambda\neq 0,$ then $\lambda u=\sum_{i=1}^n(\lambda x_i)\otimes y_i$ is a representation of $\lambda u$; therefore, we have the following: 
\begin{equation}\label{eq106}
\begin{split}
\PNo{\lambda u}\le \max_i \Noh{\lambda x_i}\, &\Nok{y_i}=\max_i \abs{\lambda}\Noh{x_i}\,\Nok{y_i}=\abs{\lambda}\max_i \Noh{x_i}\,\Nok{y_i}\\ &\implies \PNo{\lambda u}\le \abs{\lambda} \PNo{u}.
\end{split}
\end{equation}
On the other hand, we can write $\PNo{u}=\PNo{\lambda^{-1}\lambda u}$, so that
\begin{equation}\label{eq.107}
\PNo{\lambda^{-1}\lambda u}\le \abs{\lambda}^{-1}\PNo{\lambda u}\implies \abs{\lambda} \PNo{u}\le \PNo{\lambda u}.
\end{equation}
Thus, by~\eqref{eq106} and~\eqref{eq.107} we actually have that $\PNo{\lambda u}=\abs{\lambda} \PNo{u}$. 

It remains to prove that $\PNo{\cdot}$ satisfies the strong triangle inequality. In fact, for $\epsilon>0$, let $u=\sum_{i=1}^n x_i\otimes y_i$ be a representation of $u$ and $v=\sum_{i=n+1}^{n+m} x_i\otimes y_i$ be representation of $v$ such that 
\begin{equation}
\max_{0\le i\le n} \Noh{x_i}\, \Nok{y_i}\le \PNo{u}+\epsilon \ \ \ \text{and} \ \ \ \max_{n+1\le i \le n+m} \Noh{x_i}\, \Nok{y_i}\le \PNo{v}+\epsilon.
\end{equation} 
Observe that $\sum_{i=1}^{n+m} x_i\otimes y_i$ is a representation of $u+v$ and, therefore, 
\begin{equation}
\begin{split}
\PNo{u+v}\le \max_{1 \le i\le n+m}\Noh{x_i}\, \Nok{y_i}&=\max\{\max_{0\le i\le n}\Noh{x_i}\, \Nok{y_i}, \max_{n+1\le i \le n+m}\Noh{x_i}\, \Nok{y_i}\}\\ &\le \max\{\PNo{u}+\epsilon,\PNo{v}+\epsilon\}.
\end{split}
\end{equation}
 Since this is true for every $\epsilon>0$, we have that $\PNo{u+v}\le \max\{\PNo{u},\PNo{v}\}$, as we wanted to prove.
\end{proof}

Next, to establish that $\PNo{\cdot}$ is an ultrametric norm, we still need to prove that $\PNo{x}=0$ iff $x=0$. To this end,  we start with the following:

\begin{lemma}\label{lemmaequiv}
If $\{x_i\}_{i\in I}$ is a finite norm-orthogonal system in $\Hi$, then 
\begin{equation}
\bigg\|\sum_{i\in I}x_i\otimes y_i\bigg\|_\pi= \max_{i\in I}\Noh{x_i}\,\Nok{y_i},
\end{equation}
for any choice of the set $\{y_i\}_{i\in I}\subset\Ki.$ 
\end{lemma}

\begin{proof}
First note that the inequality $\PNo{\sum_{i\in I}x_i\otimes y_i}\le \max_{i\in I}\Noh{x_i}\Nok{y_i}$ holds by the definition of the ultrametric seminorm $\PNo{\cdot}$. To prove the reverse inequality --- that is, $\max_{i\in I}\Noh{x_i}\,\Nok{y_i}\le \PNo{\sum_{i\in I}x_i\otimes y_i}$ --- we need to show that 
\begin{equation}\label{eq38}
\max_{i\in I}\Noh{x_i}\,\Nok{y_i}\le \max_{j\in J}\Noh{x'_j}\,\Nok{y'_j},
\end{equation}
where $\{x'_j\}$ and $\{y'_j\}$ is any pair of finite subsets of $\Hi$ and $\Ki$, respectively, satisfying the equation $\sum_{j\in J}x'_j\otimes y'_j=\sum_{i\in I}x_i\otimes y_i\equiv u$; then, taking the infimum over all possible representations of $u$ on the r.h.s.\ of the inequality~\eqref{eq38} yields the result. In fact, let $Y\subset \Ki$ be the linear span of the vectors $\{y_i\}_{i\in I}\cup\{y'_j\}_{j\in J}$, and fix a norm-orthogonal basis $\{e_k\}_{k\in K}$ in the finite-dimensional subspace $Y$ (which always exists, see Theorem~{2.3.22} in~\cite{perez2010locally}). Then, we have that
\begin{equation}
y_i=\sum_{k}\lambda_{ik}e_k \: \text{ and }\: y'_j=\sum_{k}\lambda'_{jk}e_k,
\end{equation}
with $\lambda_{ik},\lambda'_{jk}\in \Qe.$ By linearity of the tensor product, the following chain of equations holds:
\begin{equation}
\sum_{k\in K} \bigl(\sum_{i\in I}\lambda_{ik}x_i\bigr)\otimes e_k=\sum_{i\in I} x_i\otimes y_i=\sum_{j\in J} x'_j\otimes y'_j=\sum_{k\in K} \bigl(\sum_{j\in J}\lambda'_{jk}x'_j\bigr)\otimes e_k.
\end{equation} 
Since $\{e_k\}_{k\in K}$ is a linearly independent set, it follows that
\begin{equation}
\sum_{i\in I}\lambda_{ik}x_i=\sum_{j\in J}\lambda'_{jk}x'_j,
\end{equation}
for all $k\in K.$ Finally, by a simple estimate, we get the desired result:
\begin{align}
\max_{j\in J}\Noh{x'_j}\,\Nok{y'_j}
& =
\max_{j\in J}\Noh{x'_j}\Big(\max_{k\in K}\abs{\lambda'_{jk}}\Nok{e_k}\Big)
\nonumber \\
& \ge
\max_{k\in K} \Big(\Nok{e_k} \, \Noh{\mbox{$\sum_{j\in J}$}\lambda'_{jk}x'_j}\Big)
\notag\\
& =
\max_{k\in K}\Big( \Nok{e_k}\,\Noh{\mbox{$\sum_{i\in I}$} \lambda_{ik}x_i}\Big)
\nonumber \\
& =
\max_{k\in K} \Nok{e_k}\Big(\max_{i\in I}\abs{\lambda_{ik}}\Noh{x_i}\Big)
= \max_{i\in I} \Noh{x_i} \Nok{y_i}.
\end{align}
Here, we have used the fact that the sets $\{e_k\}_{k\in K}$ and $\{x_i\}_{i\in I}$ are norm-orthogonal systems and the strong triangle inequality.
\end{proof}

\begin{theorem}
The seminorm $\PNo{\cdot}$ is, actually, a norm on $\Hi\atp\Ki$, and therefore the pair $\Hi\ptp\Ki\equiv(\Hi\atp\Ki,\PNo{\cdot})$ is a $p$-adic normed space.
\end{theorem}

\begin{proof}
We only need to prove that $u\neq 0 \implies \PNo{u}\neq 0$. Let $0\neq u\in \Hi\atp \Ki$, and express $u$ as a finite sum of the form
\begin{equation} \label{sumsim}
u=\sum_{i\in I}x_i\otimes y_i    
\end{equation}
where the vectors $x_i\in \Hi$ and $y_i\in \Ki$ are all nonzero. Let $Y\subset \Ki$ be the linear span of the vectors $\{y_i\}_{i\in I}$, and pick a norm-orthogonal basis for $Y$.  Possibly by expanding the vectors $\{y_i\}_{i\in I}$ with respect to this basis, and next replacing the vectors $\{x_i\}_{i\in I}$ with suitable linear combinations of these vectors, we see that in~\eqref{sumsim} one can always assume that $\{y_i\}_{i\in I}$ is a norm-orthogonal system. Therefore, with this assumption, by Lemma~\ref{lemmaequiv} we conclude that $\PNo{u}= \max_{i\in I} \Noh{x_i}\Nok{y_i}>0$. 
\end{proof}

By taking the completion of the $p$-adic normed space $\Hi\ptp\Ki\equiv(\Hi\atp \Ki,\PNo{\cdot})$ w.r.t.\ the metric induced by the norm $\PNo{\cdot}$, we obtain the $p$-adic Banach space $\big(\Hi\ptpc\Ki,\PNo{\cdot}\big)$, that will be synthetically denoted by $\Hi\ptpc\Ki$.

\begin{theorem}\label{th.3.9}
If $\{x_i\}_{i\in I}$ is a norm-orthogonal system in $\Hi$ and $\{y_j\}_{j\in J}$ is a norm-orthogonal system in $\Ki,$ then $\{x_i\otimes y_j\mid i\in I,\, j\in J\}$ is a norm-orthogonal system in $\Hi\ptpc\Ki$. In particular, if $\Phi=\{\phi_i\}_{i\in I}$ and $\Psi=\{\psi_j\}_{j\in J}$ are norm-orthogonal (normal) bases in $\Hi$ and $\Ki$, respectively, then $\Xi\coloneqq\{\xi_k\}_{k\in K}\equiv\{\phi_i\otimes \psi_j\}_{i\in I,\, j\in J}$ is a norm-orthogonal (normal) basis in $\Hi\ptpc\Ki$.
\end{theorem}

\begin{proof}
To establish the norm orthogonality of $\{x_i\otimes y_j\}_{i\in I,j\in J}$,
it suffices by definition to verify it for any finite subset; hence, we may assume the index sets $I$ and $J$ to be finite.
 Take $\lambda_{ij}\in \Qe$  for $i\in I$ and $j\in J$. Using Lemma \ref{lemmaequiv} with the orthogonality of the $\{y_j\}_{j\in J}$ and the property that $\PNo{x\otimes y}=\Noh{x}\cdot \Nok{y}$ for all $x\in \Hi,y\in \Ki$ we get 
\begin{equation}
\begin{split}
\bigg\|\sum_{i\in I,\, j\in J}\lambda_{ij}x_i\otimes y_j\bigg\|_\pi&=\bigg\|\sum_{i\in I}x_i\otimes (\sum_{j\in J}\lambda_{ij}y_j)\bigg\|_\pi=\max_{i\in I} \Noh{x_i}\, \bigg\|\sum_{j\in J} \lambda_{ij}y_j\bigg\|_{\mbox{\tiny $\Ki$}} \\&=\max_{i\in I} \Noh{x_i}\max_{j\in J} \abs{\lambda_{ij}}\, \Nok{y_j}=\max_{i\in I,j\in J}\abs{\lambda_{ij}} \PNo{x_i\otimes y_j}.
\end{split}
\end{equation}
Now, since $\Phi$ and $\Psi$ are normal basis for $\Hi$ and $\Ki$ respectively, by Lemma~$10.2.12$ in~\cite{perez2010locally} we have that the linear span of $\Xi$ is dense in $\Hi\ptpc\Ki$. Finally, by Theorem~$9.2.1$ in~\cite{perez2010locally} $\Xi$ is therefore a norm-orthogonal basis.
\end{proof}

Eventually, we can introduce a suitable inner product on the $p$-adic Banach space $\Hi\ptpc\Ki$, in such a way as to obtain a $p$-adic Hilbert space.

\begin{definition} \label{defsp}
We get a non-degenerate, Hermitian sesquilinear form $\inprd_{\alpha}$ on $\Hi\atp\Ki$ extending --- by linearity in the second argument and by anti-linearity
in the first argument --- the following definition:
\begin{equation}\label{inner}
\braket{x_1\otimes y_1}{x_2 \otimes y_2}_{\alpha}\coloneqq\braket{x_1}{x_2}_{\Hi}\,\braket{y_1}{y_2}_{\Ki}, \quad \forall x_1,x_2\in \Hi,y_1,y_2\in \Ki.
\end{equation}
\end{definition}

The previous definition rests on the following fact:

\begin{proposition}\label{prop.3.13}
There is a unique non-degenerate, Hermitian sesquilinear form $\inprd_{\alpha}$ on $\Hi\atp\Ki$ determined by~\eqref{inner}. In fact, we have that
\begin{enumerate}[label=\sf{(\roman*)}]

\item \label{wdef}
The quantity $\braket{u}{v}_{\alpha}$ is well defined, that is, when computed according to Definition~\ref{defsp}, it does not depend on the particular decompositions $u=\sum_{i\in I} x_i\otimes y_i$ and $v=\sum_{j\in J} x'_j\otimes y'_j$ of elements $u,v\in \Hi \atp \Ki$.

\item \label{sesq}
The mapping $\Hi\atp\Ki\ni (u,v)\mapsto\braket{u}{v}_{\alpha}\in\Qe$ is a Hermitian sesquilinear form.

\item \label{nondeg}
The sesquilinear form $\inprd_{\alpha}$ is non-degenerate.

\item \label{causch}
The Cauchy-Schwarz inequality holds, i.e.,
\begin{equation}\label{c-s}
\abs{\braket{u}{v}_{\alpha}}\leq\Noh{u}\,\Nok{v}, \quad \forall u,v\in \Hi \atp \Ki.
\end{equation}

\end{enumerate}
\end{proposition}

\begin{proof}
To prove~\ref{wdef}, it is sufficient to argue that, if $\sum_{j\in J}x_j'\otimes y_j'$ is any representation of the \emph{zero} tensor $0\in\Hi\atp\Ki$, then, for any $\sum_{i\in I} x_i\otimes y_i\in \Hi\atp\Ki$, by~\eqref{inner} we have: 
\begin{equation} \label{tool}
\begin{split}
\sum_{i\in I,\, j\in J}\braket{x_i\otimes y_i}{x_j'\otimes y_j'}_{\alpha}
& = \sum_{i\in I}\sum_{j\in J}
\braket{x_i}{x_j'}_{\Hi}\,\braket{y_i}{y_j'}_{\Ki}
\\
& = 0 = \sum_{i\in I,\, j\in J}\braket{x_j'\otimes y_j'}{x_i\otimes y_i}_{\alpha}.
\end{split}
\end{equation}
Here, the second equality follows from the implication \ref{vii}~$\implies$~\ref{viii} in Proposition~\ref{proptens} (by considering the duals of the Hilbert spaces $\Hi$ and $\Ki$). At this point, by~\eqref{tool}, one easily checks that~\ref{wdef} holds true.

The proof of~\ref{sesq} is clear, by applying Definition~\ref{defsp}.

Let us now prove~\ref{nondeg}. Given $v=\sum_{j\in J}x_j'\otimes y_j'\in \Hi\atp\Ki$, if $\braket{u}{v}_{\Ki}=0$, for every simple tensor $u=x\otimes y\in \Hi\atp\Ki$, then
\begin{equation}
\begin{split}
0 & =\braket{u}{v}_{\Ki}
\\
& = \sum_{j\in J}\braket{x\otimes y}{x_j'\otimes y_j'}_{\alpha}
\\
& = \sum_{j\in J}\braket{x}{x_j'}_{\Hi}\,\braket{y}{y_j'}_{\Ki}, \quad
\forall x\in\Hi,\ \forall y\in\Ki ;
\end{split}
\end{equation}
hence, by the implication \ref{viii}~$\implies$~\ref{vii} in Proposition~\ref{proptens}, we conclude that $v=0$. It follows that the sesquilinear form $\inprd_{\alpha}$ is non-degenerate.

Let us finally prove~\ref{causch}. In fact, for every pair of vectors $u=\sum_{i\in I} x_i\otimes y_i$, $v=\sum_{j\in J}x_j'\otimes y_j'$ in the algebraic tensor product $\Hi\atp\Ki$, we have:
\begin{align}
\abs{\braket{u}{v}_{\alpha}}=\abs{\sum_{i\in I}\sum_{j\in J}\langle  x_i\otimes y_i,x_j'\otimes y_j'\rangle_{\alpha}}
& =
\abs{\sum_{i\in I}\sum_{j\in J}\braket{x_i}{x_j'}_{\Hi}\,\braket{y_i}{y_j'}_{\Ki}} \notag \\
& \le
\max_{i\in I,\, j\in J}\abs{\braket{x_i}{x_j'}_{\Hi}}\abs{\braket{y_i}{y_j'}_{\Ki}} \notag \\
&\le
\max_{i\in I,\, j\in J}\Noh{x_i}\Nok{y_i}\Noh{x_j'}\Nok{y_j'}=\PNo{u}\PNo{v} .
\end{align}
Here, the second inequality has been obtained by the Cauchy-Schwarz inequality in the $p$-adic Hilbert spaces $\Hi$ and $\Ki$.
        
The proof is complete.     
\end{proof}

\begin{remark}
By Proposition~\ref{prop.3.13}, $\inprd_{\alpha}$ is a well-defined non-Archimedean inner product on the normed algebraic tensor product $\Hi\ptp\Ki\equiv(\Hi\atp\Ki,\PNo{\cdot})$. In particular, it is bounded w.r.t.\ the ultrametric norm $\PNo{\cdot}$, by virtue of the Cauchy-Schwarz inequality~\eqref{c-s}.
\end{remark}

\begin{lemma} \label{leminn}
The inner product $\inprd_{\alpha}$ on the normed algebraic tensor product $(\Hi\atp\Ki,\PNo{\cdot})$ can be uniquely extended to a non-Archimedean inner product $\inprd_{\pi}$ on the $p$-adic Banach space $\Hi\ptpc\Ki$. Precisely, for every pair of vectors $u,v\in\Hi\ptpc\Ki$, we have that
\begin{equation} \label{definn}
\braket{u}{v}_{\pi}=\lim_{n\to\infty}\braket{u_n}{v_n}_{\alpha},
\end{equation}
where $\{u_n\}_{n\in \N}$ and $\{v_n\}_{n\in \N}$ are any two sequences in $\Hi\atp\Ki$ such that $u_n\to u$ and $v_n\to v$ (convergence w.r.t.\ the norm $\PNo{\cdot}$ being understood).
\end{lemma}

\begin{proof}
Given any pair of vectors $u,v\in\Hi\ptpc\Ki$, let $\{u_n\}_{n\in \N}$ and $\{v_n\}_{n\in \N}$ be two sequences in $\Hi\atp\Ki$ such that $u_n\to u$ and $v_n\to v$ (w.r.t.\ the norm $\PNo{\cdot}$). We first prove that the limit $\braket{u}{v}_{\pi}\equiv\lim_{n\to\infty}\braket{u_n}{v_n}_{\alpha}$ exists and does not depend on the sequences that we have chosen. In fact, since 
\begin{equation}
|\braket{u_n}{v_n}_{\alpha}-\braket{u_m}{v_m}_{\alpha}|\le\PNo{u_n-u_m}\PNo{v_n}+\PNo{u_m}\PNo{v_n-v_m},
\end{equation}
we have that $\{\braket{u_n}{v_n}_{\alpha}\}_{n\in \N}$ is a Cauchy sequence and, therefore, is convergent. To prove that its limit is unique, let us take any other pair of sequences $\{u_n'\}_{n\in \N}$ and $\{v_n'\}_{n\in \N}$ such that $u_n'\to u$ and $v_n'\to v$. We have:
\begin{equation}
|\braket{u_n}{v_n}_{\alpha}-\braket{u_n'}{v_n'}_{\alpha}|\le\PNo{u_n-u_n'}\PNo{v_n}+\PNo{u_n'}\PNo{v_n-v_n'}\to 0.
\end{equation}

By the previous two points, for every pair of vectors $u,v\in\Hi\atp\Ki$, considering the sequences $\{u,u,\ldots\}$ and $\{v,v,\ldots\}$, we see that $\braket{u}{v}_{\alpha}=\braket{u}{v}_{\pi}$.

Moreover, one can easily check that the mapping $\Hi\ptpc\Ki\ni u,v\mapsto\braket{u}{v}_{\pi}\in\Qe$ is a Hermitian sesquilinear form, that, by the previous point, extends $\inprd_{\alpha}$ to the Banach space completion $(\Hi\ptpc\Ki,\PNo{\cdot})$ of the normed algebraic tensor product $(\Hi\atp\Ki,\PNo{\cdot})$.

Finally, let us observe that
\begin{equation}
\abs{\braket{u}{v}_{\pi}}=\lim_{n\to\infty}\abs{\braket{u_n}{v_n}_{\alpha}}\le\lim_{n\to\infty}\PNo{u_n}\PNo{v_n}=\PNo{u}\PNo{v},
\end{equation}
which proves that the Hermitian sesquilinear form $\inprd_{\pi}$ satisfies the Cauchy-Schwarz inequality too, hence, it is a non-Archimedean inner product on the $p$-adic Banach space $\Hi\ptpc\Ki$. Clearly, every sesquilinear form, that extends $\inprd_{\alpha}$ to the Banach space completion of the normed algebraic tensor product $(\Hi\atp\Ki,\PNo{\cdot})$ and satisfies the Cauchy-Schwarz inequality, is continuous and, therefore, is uniquely determined by a condition of the form~\eqref{definn}; otherwise stated, it must coincide with $\inprd_{\pi}$.
\end{proof}

\begin{theorem} \label{maintheo}
The quadruple $\big(\Hi\ptpc\Ki,\PNo{\cdot},\inprd_{\pi},\Xi\equiv\{\xi_k\}_{k\in K}\equiv\{\phi_i\otimes \psi_j\}_{i\in I,\, j\in J}\big)$ --- which, in the following, will be synthetically denoted by $\Hi\otimes\Ki$ --- is a $p$-adic Hilbert space.
\end{theorem}

\begin{proof}
The pair $(\Hi\ptpc\Ki,\PNo{\cdot})$ is a $p$-adic Banach space that, further endowed with the non-Archimedean inner product $\inprd_{\pi}$ defined by relation~\eqref{definn} (Lemma~\ref{leminn}), becomes a inner product $p$-adic Banach space. Finally, to conclude our construction, we observe that the norm-orthogonal basis $\Xi\equiv\{\xi_k\}_{k\in K}\equiv\{\phi_i\otimes\psi_j\}_{i\in I,\, j\in J}$ in $\Hi\ptpc\Ki$ is also IP-orthogonal w.r.t.\ the inner product $\inprd_{\pi}$, and, moreover, $\langle e_k,e_l\rangle_{\pi}=\delta_{kl}$.
\end{proof}

\begin{remark}
It should be clear that the synthetic symbol $\Hi\otimes\Ki$ denoting \emph{the tensor product of the $p$-adic Hilbert spaces $\Hi$ and $\Ki$} (i.e., a $p$-adic Hilbert space itself) mimics the usual symbol $\mathscr{H}\otimes\mathscr{K}$ for the tensor product  of two \emph{complex} Hilbert spaces $\mathscr{H}$ and $\mathscr{K}$.
\end{remark}

\begin{corollary}
Given orthonormal bases $\{\phi_i\}_{i\in I}$ and $\{\psi_j\}_{j\in J}$ in $\Hi$ and $\Ki$, respectively, let $\{\phi_i\otimes\psi_j\}_{i\in I,\, j\in J}$ be the associated orthonormal basis in the $p$-adic Hilbert space $\Hi\otimes\Ki$. Then, for every vector $u=\sum_{i\in I,\, j\in J}\lambda_{ij}\, \phi_i\otimes\psi_j\in\Hi\otimes\Ki$, we have that $\lambda_{ij}=\langle\phi_i\otimes\psi_j, u\rangle_{\pi}$ and
\begin{equation}
\PNo{u}=\max_{i\in I,\, j\in J} |\lambda_{ij}|.
\end{equation}
\end{corollary}

\begin{proof}
This result is directly derived from Theorem~\ref{maintheo} and the non-Archimedean Parseval identity in~\eqref{eq.11n}.
\end{proof}


\section{Relation with \texorpdfstring{$p$}{p}-adic Hilbert-Schmidt-class operators}
\label{sec.4}

As is well known, there exists an Hilbert space isomorphism between the complex tensor product Hilbert space $\mathscr{H} \otimes\mathscr{K}$, and the space of Hilbert-Schmidt operators from $\mathscr{H}$ to $\mathscr{K}$ \cite{kadi2005,bour2002,gudder2020operator}. Here we show that this picture admits a natural extension to the $p$-adic setting; namely, we prove that there exists an inner-product preserving surjective isometry --- i.e., an isomorphism of $p$-adic Hilbert spaces --- between $\T(\Hi,\Ki)$ and $\Hi \otimes \Ki$. However, as a preliminary step, we need to first introduce a suitable notion of antiunitary operator in a $p$-adic Hilbert space.
\subsection{Anti-unitary operators}\label{subsec.2.3}
We start with the following
\begin{definition}
Let $\Hi$ be a $p$-adic Hilbert space. We say that an operator $A\colon\Hi\dashrightarrow \Hi$
is \emph{conjugate-linear} (or, \emph{anti-linear}), if the relation
 \begin{equation}\label{eq.59} A(\alpha\psi+\beta\phi)=\overline{\alpha}A\phi+\overline{\beta}A\psi
 \end{equation}
holds for every $\alpha,\beta\in\Qe$, and $\phi,\psi\in\dom(A)$.
\end{definition}
Similarly to linear operators, a conjugate-linear operator $A\colon \Hi\dashrightarrow\Hi$ is  
continuous precisely when it is bounded; moreover, we say that $A$ is all-over if $\dom(A)=\Hi$.

\begin{definition}\label{def.2.13}
Let $\Phi\equiv\{\phi_m\}_{m\in\N}$ be an orthonormal basis in a $p$-adic Hilbert space $\Hi$. We call an all-over, bounded, conjugate-linear operator $J_\Phi\colon\Hi\rightarrow\Hi$, satisfying the condition 
\begin{equation}\label{eq.61n}
J_\Phi\,\phi_m=\phi_m,\quad\forall m\in\N, 
\end{equation}
a \emph{conjugation operator} associated with the orthonormal basis $\Phi$.
\end{definition}

\begin{remark}\label{rem.2.17n}
Every orthonormal basis $\Phi\equiv\{\phi_m\}_{m\in\N}$ in $\Hi$ can be associated, in a natural way, to a  conjugation operator $J_\Phi$ (see the discussion following Example~\ref{exa.2.24}). Moreover, it is not difficult to check that if $L$ is a linear operator in $\Hi$, the operator defined as $A\coloneqq J_\Phi L$ (or, equivalently, $A\coloneqq LJ_\Phi$), is a conjugate-linear operator and, conversely, any conjugate-linear operator $A$ in $\Hi$ can be expressed in this form, for some linear operator $L$ on $\Hi$.
\end{remark}
\begin{remark}\label{rem.4.4}
We recall that a linear operator $L\colon\Hi\dashrightarrow \Hi$ is \emph{norm-preserving} (N-preserving, in short), if it satisfies the condition: $\No{L\phi}=\No{\phi}$, for any $\phi\in\dom(L)$~\cite{aniello2023trace}.
An all-over N-preserving map in $\Hi$ is a linear \emph{isometry} of $\Hi$. A linear operator $L$ in $\Hi$ is \emph{inner-product-preserving} (IP-preserving), if 
\begin{equation}
 \braket{L\phi}{L\psi}=\braket{\phi}{\psi}, \quad \forall \phi,\psi\in\dom(L),   
\end{equation}
while $L$ is \emph{norm-orthogonality-preserving} (NO-preserving), if it is linear, and satisfies the condition~\cite{aniello2023trace}:
\begin{equation}
\phi,\psi\in\dom(L),\quad \phi\nperp\psi\;\implies\;L\phi\nperp L\psi.
\end{equation}
An IP-preserving surjective isometry (i.e., an automorphism) of $\Hi$ is a \emph{unitary operator}; adopting the notation already used in~\cite{aniello2023trace}, in the following we will denote the set of all such operators by $\U$. (For a complete characterization of unitary operators in a $p$-adic Hilbert space see, in particular, Theorem~$5.13$ in~\cite{aniello2023trace}). 
\end{remark}
Bearing in mind Remarks~\ref{rem.2.17n} and~\ref{rem.4.4}, we are now led to set the following
\begin{definition}
Let $A\colon\Hi\dashrightarrow \Hi$ be a conjugate-linear operator in a $p$-adic Hilbert space $\Hi$. We say that
\begin{enumerate}[label=\tt{(A\arabic*)}]
\item $A$ is an \emph{anti-isometric} operator, if $A$ is all-over and N-preserving --- i.e., if $\No{A\phi}=\No{\phi}$ for every  $\phi\in\dom(A)=\Hi$.

\item $A$ is a \emph{conjugate-linear topological isomorphism}, if $A$ is bounded and admits a (conjugate-linear) bounded inverse. 

\item $A$ is \emph{inner-product-conjugating}, if the condition
\begin{equation}
    \braket{A\phi}{A\psi}=\overline{\braket{\phi}{\psi}}
\end{equation}
is satisfied for every $\phi,\psi$ in $\dom(A)$.
\end{enumerate}
\end{definition}
\begin{proposition}\label{prop.2.14}
Let $\Hi$ be a $p$-adic Hilbert space, let $\Phi\equiv\{\phi_m\}_{m\in\N}$ be an orthonormal basis, and let $J_\Phi$ be the associated conjugation operator. Then, the following facts hold:
\begin{enumerate}[label=(\roman*)]
\item \label{prop.2.14_i} For every vector $\chi$ in $\Hi$, we have:
\begin{equation}\label{eq.61}
    J_\Phi\chi=J_\Phi\bigg(\sum_n\braket{\phi_n}{\chi}\phi_n\bigg)=\sum_n\overline{\braket{\phi_n}{\chi}}\phi_n.
\end{equation}
Moreover, $J_\Phi$ is IP-conjugating and anti-isometric.
\item \label{prop.2.14_ii} $J_\Phi$ is involutive, i.e., $J_\Phi^2=\Id$, bijective, bounded, and adjointable. In particular, $J_\Phi=J_\Phi^{-1}=J_\Phi^\ast$.
\end{enumerate}
\end{proposition}
\begin{proof}
 Relation~\eqref{eq.61} simply follows by anti-linearity from Definition~\ref{def.2.13}. Then,  we have:
\begin{equation}
\braket{J_\Phi\chi}{J_\Phi\eta}=\Big\langle\sum_n\overline{\braket{\phi_n}{\chi}}\phi_n,\,\sum_m\overline{\braket{\phi_m}{\eta}}\phi_m\Big\rangle=\sum_n\braket{\phi_n}{\chi}\overline{\braket{\phi_m}{\eta}}=\overline{\braket{\chi}{\eta}},
\end{equation}
i.e., $J_\Phi$ is IP-conjugating. Moreover,
by observing that
\begin{equation}
\No{J_\Phi\chi}=\No{\mbox{$\sum_n\overline{\braket{\phi_n}{\chi}}\phi_n$}}=\sup_n|\overline{\braket{\phi_n}{\chi}}|=\sup_n|\braket{\phi_n}{\chi}|=\No{\chi},
\end{equation}
and since, from~\eqref{eq.61}, $\dom(J_\Phi)=\Hi$, we see that $J_\Phi$ is an anti-isometric operator, i.e.,~\ref{prop.2.14_i} holds. Let us now prove
\ref{prop.2.14_ii}. Exploiting~\eqref{eq.61} twice, we have:
\begin{equation}
J_\Phi^2\chi=J_\Phi(J_\Phi\chi)=J_\Phi\bigg(\sum_n\overline{\braket{\phi_n}{\chi}}\phi_n\bigg)=\sum_n\braket{\phi_n}{\chi}\phi_n,
\end{equation}
 i.e., $J_\Phi^2=\Id$. Now, it is clear that $J_\Phi$ is surjective and, by point~\ref{prop.2.14_i}, an anti-isometric operator. Therefore, $J_\Phi$ is bijective, bounded, and $J_\Phi^{-1}=J_\Phi$. Moreover, by observing that
  \begin{equation}
  \braket{\phi_n}{J_\Phi^\ast\phi_m}=\overline{\braket{J_\Phi^{-1}\phi_n}{J_\Phi^\ast J^{-1}_\Phi\phi_m}}=\overline{\braket{J_{\Phi}J_\Phi^{-1}\phi_n}{J_\Phi^{-1}\phi_m}}=\overline{\braket{\phi_n}{\phi_m}}=\delta_{nm}, 
  \end{equation}
  we also see that $J_\Phi^\ast\phi_m=\phi_m$, for all $m\in\N$, i.e., $J_\Phi^{\ast}=J_\Phi$.
\end{proof}
As pointed out in Remark~\ref{rem.2.17n}, composing a linear operator $L$ in $\Hi$ with the conjugation operator $J_\Phi$ --- associated with the orthonormal basis $\Phi\equiv\{\phi_m\}_{m\in\N}\subset \Hi$ --- yields a conjugate-linear operator in $\Hi$;  conversely, every conjugate-linear operator in $\Hi$ can be expressed in this form. It is precisely this observation which motivates our next
\begin{definition}
Let $\Hi$ be a $p$-adic Hilbert space, and let $\Phi\equiv\{\phi_m\}_{m\in\N}$ be an orthonormal basis. We say that   $Z\colon\Hi\rightarrow \Hi$ is an \emph{anti-unitary operator}, if it can be expressed in the form
\begin{equation}
Z=J_\Phi U,\quad U\in\U,
\end{equation}
i.e., as the composition of the conjugation operator $J_\Phi$ --- associated with the orthonormal basis $\Phi\equiv\{\phi_m\}_{m\in\N}$ in $\Hi$ --- and some unitary operator $U\in\U$. We will denote the set of all the anti-unitary operators in $\Hi$ by $\AU$.
\end{definition}
Next result provides the conjugate-linear counterpart of Theorem~5.13 in~\cite{aniello2023trace}.

\begin{theorem}\label{th.2.17}
Given a conjugate-linear operator $Z$ in $\Hi$, the following facts are equivalent:
\begin{enumerate}[label=\tt{(Z\arabic*)}]
\item \label{th.2.17_w1} $Z$ is an anti-unitary operator, i.e., $Z=J_\Psi U$, for $J_\Psi$ the conjugation operator associated with the orthonormal basis $\Psi\equiv\{\psi_n\}_{n\in\N}\subset\Hi$, and $U$ some unitary operator in $\U$;

\item \label{th.2.17_w2} $Z$ is bounded and, for some pair of orthonormal bases $\Phi\equiv\{\phi_m\}_{m\in\N}$ and $\Psi\equiv\{\psi_n\}_{n\in\N}$ in $\Hi$, we have that $Z\phi_k=\psi_k$, $\forall k\in \N$;

\item \label{th.2.17_w3} $Z$ is bounded and adjointable,
 with $ZZ^\ast=Z^\ast Z=\Id$,  and $\No{Z}=1$;

\item \label{th.2.17_w4} $Z$ is a surjective, IP-conjugating, all-over operator, and $\No{Z}=1$;

\item \label{th.2.17_w5} $Z$ is an IP-conjugating, conjugate-linear topological isomorphism, and $\No{Z}=1=\No{Z^{-1}}$;

\item \label{th.2.17_w6} $Z$ is an anti-automorphism of the $p$-adic Hilbert space $\Hi$, i.e., an IP-conjugating surjective anti-isometry;

\item \label{th.2.17_w7} $Z$ is a surjective, IP-conjugating, NO-preserving, all-over operator;

\item \label{th.2.17_w8} $Z$ is bounded and transforms orthonormal bases into orthonormal bases. 

\end{enumerate}
\end{theorem}
\begin{proof}
Let us first prove~\ref{th.2.17_w1} $\iff$ \ref{th.2.17_w2} $\implies$ \ref{th.2.17_w3}. Indeed, let $Z\in\AU$ and let $\Phi\equiv\{\phi_m\}_{m\in\N}$ and $\Psi\equiv\{\psi_n\}_{n\in\N}$ be two orthonormal bases in $\Hi$. Let $J_\Psi$ be the conjugation operator associated with $\Psi$ --- i.e., $J_\Psi\psi_n=\psi_n$, $\forall n\in\N$. Since $Z$ is anti-unitary, there exists a unitary operator $U$ in $\U$ such that $Z=J_\Psi U$. In particular, we can choose $U$ in such a way that
\begin{equation}
Z\phi_k=J_\Psi U\phi_k=J_\Psi\psi_k=\psi_k,\quad \forall k\in\N,
\end{equation}
i.e., as the unitary operator relating the orthonormal bases $\Phi$ and $\Psi$ in $\Hi$
(cf.\ property {\tt (U$2$)} of Theorem~$5.13$ in~\cite{aniello2023trace}). Moreover, it is clear that $Z=J_\Psi U$ is bounded; hence,~\ref{th.2.17_w1} $\implies$~\ref{th.2.17_w2}. Conversely, let $Z$ be as in~\ref{th.2.17_w2}. Since it is conjugate-linear, we can express $Z$ as $Z=J_\Psi L$, where $J_\Psi$ is the conjugation operator associated with the orthonormal basis $\Psi\equiv\{\psi_n\}_{n\in\N}$,  and $L$ is some linear operator in $\Hi$ (see Remark~\ref{rem.2.17n}). Now, since, by assumption, $Z$ is a bounded operator, $L$ is bounded too. Moreover, by observing that
\begin{equation}
J_\Psi\psi_k=\psi_k=Z\phi_k=J_\Psi L\phi_k,\quad\forall k\in\N, 
\end{equation}
we see that $L\phi_k=\psi_k$, for all $k\in\N$. Hence, owing again to~{\tt (U$2$)} of Theorem~$5.13$ in~\cite{aniello2023trace}, $L$ is unitary and, therefore, $Z$ is an anti-unitary operator. We next prove~\ref{th.2.17_w2} $\implies$~\ref{th.2.17_w3}. We have already shown that an anti-unitary operator is completely characterized by condition~\ref{th.2.17_w2}. Now, we have that
\begin{equation}\label{eq.73n}
    \braket{\phi_m}{Z^\ast\psi_k}=\braket{Z\phi_m}{\psi_k}=\braket{J_\Psi U\phi_m}{\psi_k}=\braket{J_\Psi\psi_m}{\psi_k}=\delta_{mk},
\end{equation}
where we have expressed $Z$ as $Z=J_\Psi U$, and used the fact that we can choose $U$ as the unitary operator which maps the orthonormal basis $\Phi=\{\phi_m\}_{m\in\N}$ into $\Psi\equiv\{\psi_n\}_{n\in\N}$. Therefore, from~\eqref{eq.73n}, it follows that $Z$ is adjointable, and $Z^\ast\psi_k=\phi_k$, $\forall k\in\N$, i.e., $Z^\ast$ is anti-unitary. This also entails that
\begin{equation}\label{eq.74n}
Z^\ast Z=\Id=Z Z^\ast.
\end{equation}
Moreover, since $Z=J_\Psi U$, and $J_\Psi$ is an anti-isometric operator, we observe that  
\begin{equation}\label{eq.76}
\No{Z}=\No{J_\Psi U}=\sup_{\chi\neq 0}\frac{\No{J_\Psi U\chi}}{\No{\chi}}=\sup_{\chi\neq 0}\frac{\No{U\chi}}{\No{\chi}}=1,
\end{equation}
and, hence, 
\begin{equation}
\No{Z}=\No{Z^\ast}=1.
\end{equation}
The implication~\ref{th.2.17_w3} $\implies$~\ref{th.2.17_w4} is clear since, if $Z$ is bounded, adjointable, and satisfies condition~\eqref{eq.74n}, then $Z$ is surjective and all-over. Moreover, if $\Psi\equiv\{\psi_n\}_{n\in\N}$ is an orthonormal basis in $\Hi$, by observing that
\begin{align}
\braket{Z\chi}{Z\eta}&=\Big\langle Z\sum_n\braket{\psi_n}{\chi}\psi_n,\sum_m\braket{\psi_m}{\eta}\psi_m\Big\rangle\nonumber\\
&=\Big\langle\sum_n\overline{\braket{\chi}{\psi_n}}Z\psi_n, \sum_m\overline{\braket{\psi_m}{\eta}}Z\psi_m\Big\rangle\nonumber\\
&=\sum_n\sum_m\braket{\psi_n}{\chi}\overline{\braket{\psi_m}{\eta}}\braket{Z\psi_n}{Z\psi_m}\nonumber\\
&=\sum_n\sum_m\braket{\psi_n}{\chi}\overline{\braket{\psi_m}{\eta}}\braket{Z^\ast Z\psi_n}{\psi_m}\nonumber\\
&=\sum_n\braket{\psi_n}{\chi}\overline{\braket{\psi_n}{\eta}}=\overline{\braket{\chi}{\eta}},
\end{align}
for all $\chi,\eta$ in $\Hi$, we see that $Z$ is IP-conjugating too and, hence,~\ref{th.2.17_w4} is verified.
 Next, let us prove~\ref{th.2.17_w4} $\implies$~\ref{th.2.17_w5}. Expressing $Z$ as $Z=J_\Psi L$ --- for a suitable linear operator $L$ in $\Hi$, and with $J_\Psi$ the conjugation operator associated with the orthonormal basis $\Psi\equiv\{\psi_n\}_{n\in\N}$ --- since, by assumption, $Z$ is surjective and IP-conjugating, we see that $L$ must be surjective and IP-preserving. Then, Theorem~$5.8$ in~\cite{aniello2023trace} allows us to infer that $L$ is an adjointable top-linear isomorphism such that $L^\ast=L^{-1}$ and $\No{L}=\No{L^\ast}=\No{L^{-1}}$. Hence, we have:
\begin{equation}\label{eq.78}
L^\ast J_\Psi=L^{-1} J_\Psi\;\;\implies\;\;(J_\Psi L)^{\ast}=(J_\Psi L)^{-1},
\end{equation}
i.e., $Z^\ast=Z^{-1}$. (Note that, in the second equality in~\eqref{eq.78}, we have used property~\ref{prop.2.14_ii} of Proposition~\ref{prop.2.14}). Hence, $Z$ is an adjointable conjugate-linear topological isomorphism and
\begin{equation}
\No{Z}=\No{Z^\ast}=\No{Z^{-1}}=1.
\end{equation}
Now, if $Z$ satisfies~\ref{th.2.17_w5}, expressing again $Z$ as $Z=J_\Psi L$, for some linear operator $L$ on $\Hi$, we see that, by hypothesis, $L$ must be an IP-preserving top-linear isomorphism such that $\No{L}=\No{L^{-1}}$. Then, by Lemma~$5.2$ in~\cite{aniello2023trace}, $L$ is an IP-reserving surjective isometry. Therefore, $Z=J_\Psi L$ is an IP-conjugating surjective anti-isometry, i.e.,~\ref{th.2.17_w5} $\implies$~\ref{th.2.17_w6}. Let us now show~\ref{th.2.17_w6} $\iff$~\ref{th.2.17_w7}. Indeed, if $Z=J_\Psi L$ is anti-isometric, then $L$ is an isometry, and, by Theorem~$5.11$ in~\cite{aniello2023trace}, a NO-preserving all-over operator. Recalling that $J_\Psi$ is an anti-isometry (see~\ref{prop.2.14_i} in Proposition~\ref{prop.2.14}), we see that $Z$ is an IP-conjugating, surjective, NO-preserving, all-over operator, i.e.,~\ref{th.2.17_w7} holds. Conversely, assume that $Z=J_\Psi L$ is IP-conjugating, surjective and all-over. Then, $L$ is a surjective IP-preserving operator and, hence, $\No{L}=\No{L^{-1}}$ (cf.\ Theorem~$5.8$ in~\cite{aniello2023trace}). From this, it follows that
\begin{equation}\label{eq.cond.norm}
\No{Z}=\No{J_\Psi L}=\No{L}=\No{L^{-1}}=\No{L^{-1}J_\Psi^{-1}}=\No{Z^{-1}}.
\end{equation}
Next, let us further assume that $Z$ is NO-preserving. Since $J_\Psi$ is an anti-isometry, it must be true that $L$ is an all-over and NO-preserving linear operator; then, by Theorem~$5.11$ in~\cite{aniello2023trace}, $L$ must be a non-zero scalar multiple of an isometry, i.e., $L=zI$, where $z\in\Qe\setminus\{0\}$, and $I\colon\Hi\rightarrow\Hi$ is an isometry in $\Hi$. We thus have that $Z=J_\Psi zI=\overline{z} J_\Psi I$. From this, and owing to~\eqref{eq.cond.norm}, we see that 
\begin{equation}
    \abs{z}=\abs{\overline{z}}=\No{J_\Psi Iz}=\No{Z}=\No{Z^{-1}}=\No{\overline{z}^{-1}I^{-1}J_\Psi^{-1}}=\abs{\overline{z}}^{-1}=\abs{z}^{-1},
\end{equation}
i.e., $\abs{z}=1$ and $Z=\overline{z}J_\Psi I$ is a surjective IP-conjugating anti-isometry, i.e.,~\ref{th.2.17_w6} holds. To complete the proof, it is left to show that~\ref{th.2.17_w6} $\implies$~\ref{th.2.17_w8} and~\ref{th.2.17_w8} $\implies$~\ref{th.2.17_w2}. That~\ref{th.2.17_w8} $\implies$~\ref{th.2.17_w2} is clear. Hence, let us prove~\ref{th.2.17_w6} $\implies$~\ref{th.2.17_w8}. Assume that $Z$ satisfies conditions in~\ref{th.2.17_w6}, i.e., $Z$ is an IP-conjugating surjective anti-isometry. Now, given any orthonormal basis $\Phi\equiv\{\phi_m\}_{m\in\N}$ in $\Hi$, it is readily seen that the set $\{Z\phi_n\}_{n\in\N}$ provides another orthonormal basis in $\Hi$. Indeed, since $W$ is IP-conjugating, we have:
\begin{equation}
    \braket{Z\phi_m}{Z\phi_n}=\overline{\braket{\phi_n}{\phi_m}}=\delta_{mn},
\end{equation}
i.e., $\{Z\phi_n\}_{n\in\N}$ is an IP-orthonormal set;
moreover, for every finite or denumerable set $\{\alpha_n\}_{n\in\N}\subset\Qe$, and every $\eta$ in $\Hi$, we have:
\begin{equation}
    \No{\mbox{$\sum_n z_n Z\phi_m$}}=\No{\mbox{$Z\sum_n \overline{z}_n\phi_m$}}=\max_n \abs{\overline{z}_n}=\max_n\abs{z_n},
\end{equation}
namely, $\{Z\phi_n\}_{n\in\N}$ is a normal set. 
Eventually by observing that 
\begin{align}
\chi=Z(Z^{-1}\chi)&=Z\sum_n \braket{\phi_n}{Z^{-1}\chi}\phi_n=\sum_n\overline{\braket{\phi_n}{Z^{-1}\chi}}Z\phi_n=\nonumber\\
&=\sum_n\braket{Z\phi_n}{Z^{-1}Z\chi}Z\phi_n=\sum_n\braket{Z\phi_n}{\chi}Z\phi_n,
\end{align}
we see that $\{Z\phi_n\}_{n\in\N}$ is dense in $\Hi$ and, hence, an orthonormal basis. 
This then proves that~\ref{th.2.17_w6} $\implies$~\ref{th.2.17_w8} and, therefore, concludes the proof.
\end{proof}
As is well known, in a complex Hilbert space $\mathscr{K}$, every involutive, anti-unitary operator acts as a conjugation operator w.r.t.\ some orthonormal basis in $\mathscr{K}$~\cite{barry2015operator}. Specifically, if $Z$ is anti-unitary and satisfies $Z^2=\Id$, then an orthonormal basis $\mathscr{Y}\equiv\{y_k\}_{k\in\N}$ in $\mathscr{K}$ exists such that $Zy_k=y_k$, for all $k\in\N$ --- i.e., $Z\equiv J_\mathscr{Y}$ is the conjugation operator associated with $\mathscr{Y}\equiv\{y_k\}_{k\in\N}$ (see also Remark~\ref{rem.2.23}). In the next example, we show that this picture does not extend to the $p$-adic framework, namely, that given an involutive anti-unitary operator $Z$ in a $p$-adic Hilbert space $\Hi$, one cannot always come up with an orthonormal basis $\Psi\equiv\{\psi_m\}_{m\in\N}$ in $\Hi$ preserved by $Z$.
\begin{example}\label{exa.2.24}
Let $\Hi$ be a $p$-adic Hilbert space, and let $\Phi\equiv\{\phi_m\}_{m\in\N}$ be an orthonormal basis. Let $Z\in\AU$ be such that $Z^2=\Id$ (i.e., $Z$ is an involutive anti-unitary operator in $\Hi$). Define the set of vectors $\Psi\equiv\{\psi_m\}_{m\in\N}$ in $\Hi$ by setting
\begin{equation}\label{eq.86n}
  \psi_{2m}=z_1\big(Z\phi_m+\phi_m\big),\quad \psi_{2m-1}=\frac{z_2}{\sqrt{\mu}}\big(Z\phi_m-\phi_m\big),
\end{equation}
where $z_1,z_2$ are suitably chosen coefficients in $\Q_p$. Note that, since $Z$ is anti-unitary, $\{Z\phi_m\}_{m\in\N}$ is an orthonormal set in $\Hi$ (see~\ref{th.2.17_w8} in Theorem~\ref{th.2.17}). By construction, $\Psi\equiv\{\psi_m\}_{m\in\N}$ is a $Z$-invariant dense subset of $\Hi$ --- just note that $\phi_m=\frac{1}{2}z_1\psi_{2m}-\frac{\sqrt{\mu}}{2}z_2\psi_{2m-1}$. Moreover, by suitably choosing the coefficients $z_1,z_2$ in $\Q_p$, we can turn $\Psi$ into a normal subset of $\Hi$.
Next, let us see the constraints that $z_1,z_2$ have to satisfy for $\Psi\equiv\{\psi_m\}_{m\in\N}$ to be orthonormal too. Clearly, $\braket{\psi_n}{\psi_m}=0$, for all $n\neq m\in\N$, i.e., $\Psi$ is IP-orthogonal; requiring further the IP-orthonormality we get the following conditions:
\begin{equation}\label{eq.97n}
1=\braket{\psi_m}{\psi_m}=\braket{z_1(Z\phi_m+\phi_m)}{z_1(Z\phi_m+\phi_m)}=2z_1^2,
\end{equation}
for $m$ even, while, for $m$ odd
\begin{equation}\label{eq.98n}
    1=\braket{\psi_m}{\psi_m}=\bigg\langle\frac{z_2}{\sqrt{\mu}}\big(Z\phi_m-\phi_m\big),\frac{z_2}{\sqrt{\mu}}\big(Z\phi_m-\phi_m\big)\bigg\rangle=-\frac{2z_2^2}{\mu}.
\end{equation}
Therefore,~\eqref{eq.97n} and~\eqref{eq.98n} entail that
\begin{equation}\label{eq.92n}
z_1=\sqrt{\frac{1}{2}},\quad\mbox{and},\quad z_2=\sqrt{-\frac{\mu}{2}},
\end{equation}
i.e., $z_1$, $z_2$ do not belong to $\Q_p$. On the other hand, if, in the defining condition of the set $\Psi\equiv\{\psi_m\}_{m\in\N}$ in~\eqref{eq.86n} above, we assume that $z_1,z_2$ are in $\Qe$ --- i.e., $z_1$ and $z_2$ belong to a suitable quadratic extension of $\Q_p$ --- setting $z_1$ and $z_2$ as in~\eqref{eq.92n}, we can turn $\Psi\equiv\{\psi_m\}_{m\in\N}$ into an orthonormal subset of $\Hi$; \emph{but}, in this case, it will not be an invariant set under the action of $Z$.  Concluding, we have the following dichotomy: Either $\Psi\equiv\{\psi_m\}_{m\in\N}$ is a $Z$-invariant, normal, IP-orthogonal dense subset of $\Hi$ --- in which case it is not IP-orthonormal --- or $\Psi\equiv\{\psi_m\}_{m\in\N}$ is an orthonormal dense subset, but not $Z$-invariant.
\end{example}

In example~\ref{exa.2.24}, we constructed an invariant dense subset of $\Hi$, which cannot be turned into an orthonormal set. Actually, this example is a manifestation of a peculiar feature of the conjugation operators in the $p$-adic setting, as will be clarified in the forthcoming Proposition~\ref{prop.2.22}. Before stating this result, let us first note the following. Consider the canonical Hilbert space $\CH$ endowed with the (so-called) canonical basis $\bm{e}\equiv\{e_i\}_{i\in\N}$ (see Example~\ref{exa.2.4}). Define in $\CH$ the  transformation $J_0$ by setting 
\begin{equation}
J_0\coloneqq\isoc^{-1}\con\isoc,
\end{equation}
where $\isoc\colon \CH\rightarrow c_0(\N,\Qe)$ is the natural embedding of $\CH$ into $c_0(\N,\Qe)$ given by
\begin{equation}
\CH\ni\xi=\mbox{$\sum_i \xi_ie_i$}\mapsto\isoc(\xi)\coloneqq\{\xi_i\}_{i\in\N}\in c_0(\N,\Qe),
\end{equation}
while $\con\colon c_0(\N,\Qe)\ni\{\xi_i\}_{i\in\N}\mapsto\{\overline{\xi}_i\}_{i\in\N}\in c_0(\N,\Qe)$ is the conjugate-linear isometry of $c_0(\N,\Qe)$ into itself. It is clear that $J_0$ is anti-unitary and involutive. Moreover, by observing that, for every $\xi=\sum_i\xi_ie_i$ in $\CH$,  
\begin{equation}
J_0\xi=\isoc^{-1}\con\isoc\Big(\mbox{$\sum_i \xi_ie_i$}\Big)=\isoc^{-1}(\{\overline{\xi}_i\}_{i\in\N})=\mbox{$\sum_i\overline{\xi}_ie_i$},
\end{equation}
we can argue that $J_0$ acts as the conjugation operator associated with the canonical basis $\{e_i\}_{i\in\N}$ of $\CH$ (in the following, we will refer to $J_0$ as to the \emph{canonical conjugation} in $\CH$). Let now $\Hi$ be a $p$-adic Hilbert space, and let $\Psi\equiv\{\psi_n\}_{n\in\N}$ be an orthonormal basis. Let $\ihs_\Psi$ be the canonical isomorphism of $\Hi$ into $\CH$ (see~\eqref{eq.12}). It is clear that the (conjugate-linear) operator defined as 
\begin{equation}
J_\Psi\coloneqq\ihs_\Psi^{-1} J_0 \ihs_\Psi,
\end{equation}
provides the conjugation operator associated with the orthonormal basis $\Psi\equiv\{\psi_n\}_{n\in\N}$ in $\Hi$.
Indeed, $J_\Psi$ is anti-unitary and involutive; moreover, for every $\chi=\sum_n\braket{\psi_n}{\chi}\psi_n$ in $\Hi$, we have:
\begin{align}
J_\Psi\chi&=J_\Psi\bigg(\sum_n\braket{\psi_n}{\chi}\psi_n\bigg)=\ihs_\Psi^{-1}J_0\ihs_\Psi\bigg(\sum_n\braket{\psi_n}{\chi}\psi_n\bigg)\nonumber\\
&=\ihs_\Psi^{-1}J_0\bigg(\sum_{n}\braket{\psi_n}{\chi}e_n\bigg)=\ihs_\Psi^{-1}\bigg(\sum_n\overline{\braket{\psi_n}{\chi}}J_0e_n\bigg)\nonumber\\
&=\ihs_\Psi^{-1}\bigg(\sum_n\overline{\braket{\psi_n}{\chi}}e_n\bigg)=\sum_n\overline{\braket{\psi_n}{\chi}}\psi_n,
\end{align}
i.e.,  $J_\Psi\psi_n=\ihs_\Psi^{-1}J_0\ihs_\Psi\psi_n=\psi_n$. We have thus proved that, given any orthonormal basis $\Psi\equiv\{\psi_n\}_{n\in\N}$ in $\Hi$, it is always possible to construct a conjugation operator associated with it; moreover, the so-constructed conjugation operator is equivalent, up to isomorphisms, with  the conjugation $J_0$ in the canonical Hilbert space $\CH$. A natural question to ask is whether the reverse implication holds, namely if, given an involutive anti-unitary operator $Z$ in $\Hi$, there always exists an Hilbert space isomorphism $\mathscr{I}$ such that $\mathscr{I}^{-1}Z\mathscr{I}$ reduces to the canonical conjugation operator $J_0$. The next result shows that, in fact, it is not so. 
\begin{proposition}\label{prop.2.22}
Let $Z$ be an involutive anti-unitary operator in the $p$-adic Hilbert space $\Hi$ --- i.e., $Z\in\AU$, and $Z^2=\Id$. Then, in general, there exists no isomorphism $\mathscr{I}\colon\CH\rightarrow \Hi$, such that $\mathscr{I}^{-1}Z\mathscr{I}=J_0$, where $J_0$ is the conjugation operator associated with the canonical basis $\bm{e}\equiv\{e_i\}_{i\in\N}$ in the coordinate Hilbert space $\CH$.
\end{proposition}
\begin{proof}
The existence of an isomorphism as stated in the proposition is equivalent to requiring the existence of a $Z$-invariant orthonormal basis $\Psi\equiv\{\psi_n\}_{n\in\N}$ in $\Hi$. Indeed, if an isomorphism $\mathscr{I}\colon\CH\rightarrow\Hi$ exists such that $\mathscr{I}^{-1}Z\mathscr{I}=J_0$ then, as one may easily check, the set $\{\mathscr{I}e_n\}_{n\in\N}$ --- where $\bm{e}\equiv\{e_n\}_{n\in\N}$ is the canonical basis in $\CH$ --- provides a $Z$-invariant orthonormal basis in $\Hi$. Conversely, assuming the existence of an orthonormal basis $\Psi\equiv\{\psi_n\}_{n\in\N}$ in $\Hi$ invariant under the action of $Z$, 
we have, for every $\chi$ in $\Hi$, that 
\begin{equation}
    Z\chi=Z\Big(\sum_{n}\braket{\psi_n}{\chi}\psi_n\Big)=\sum_n\overline{\braket{\psi_n}{\chi}}Z\psi_n=\sum_n\overline{\braket{\psi_n}{\chi}}\psi_n,
\end{equation}
and from this, it is clear that $\mathscr{I}\equiv W_\Psi^{-1}$ provides the isomorphism claimed in the proposition. Therefore, the proposition is proved if we can show that it is not always possible to construct a $Z$-invariant orthonormal basis in $\Hi$. To this end, let $\invset$ denote set of all the elements in $\Hi$ invariant under the action of $Z$, i.e., 
\begin{equation}
\invset\coloneqq\{\eta\in\Hi\mid Z\eta=\eta\}.
\end{equation}
Note that $\invset$ is not empty, as it contains at least the null vector; moreover, although it is not linear on $\Qe$, it is a $\Q_p$-linear space since, for any $\lambda_1,\lambda_2\in\Q_p$ and $\eta_1,\eta_2$ in $\invset$, we have that $\lambda_1\eta_1+\lambda_2\eta_2$ belongs to $\invset$ too. It is clear that $\invset$ is everywhere dense in $\Hi$: It suffices to observe that any $\chi$ in $\Hi$ can be expressed as
\begin{equation}\label{eq.100n}
    \chi=\frac{1}{2}\big(\chi_1+\sqrt{\mu}\chi_2\big),
\end{equation}
where we set 
\begin{equation}\label{eq.101n}
\chi_1\equiv Z\chi+\chi,\;\;\mbox{and}\;\; \chi_2=\frac{1}{\sqrt{\mu}}\big(Z\chi-\chi\big),
\end{equation}
with $Z\chi_1=\chi_1$ and $Z\chi_2=\chi_2$, i.e., $\chi_1,\chi_2\in\mathfrak{S}$. Thus, every $\chi$ in $\Hi$ can be expressed as a linear combination of elements in $\invset$. It then follows that a $Z$-invariant orthonormal basis in $\Hi$ exists if and only if we can construct an orthonormal basis in $\invset$. Let $\Phi\equiv\{\phi_n\}_{n\in\N}$ be an orthonormal basis in $\Hi$. A $Z$-invariant dense subset $\Psi\equiv\{\psi_n\}_{n\in\N}$ in $\Hi$ --- i.e., a dense subset in $\invset$ --- can be constructed by setting
\begin{equation}
\psi_{2n}\coloneqq Z\phi_n+\phi_n,\quad \psi_{2n-1}\coloneqq \frac{1}{\sqrt{\mu}}\big(Z\phi_n-\phi_n\big).
\end{equation}
However, we know that this set is normal and IP-orthogonal, but cannot be made orthonormal without violating the $Z$-invariance condition (see Example~\ref{exa.2.24}). Hence, in general, 
an orthonormal basis in $\invset$ --- i.e., a $Z$-invariant basis in $\Hi$ --- cannot be constructed, which is what we want to prove. 
\end{proof}
\begin{remark}\label{rem.2.23}
    A conjugation operator $J$ in a complex Hilbert space $\mathscr{K}$ is defined as an involutive anti-unitary operator. Then, it is not difficult to prove that  there exists an orthonormal basis $\Psi\equiv\{\psi_k\}_{k\in\N}$ in $\mathscr{K}$ which is invariant under the action of $J$, i.e., such that $J\psi_k=\psi_k$, for all $k\in\N$~\cite{barry2015operator}. 
    Indeed, let $\mathfrak{G}$ denote the set of all vectors $\chi$ in $\mathscr{K}$ invariant under the action of $J$, i.e., such that $J\chi=\chi$; this set is easily seen to be dense in $\mathscr{K}$. Since it is non-empty,  it is separable and, thus, contains an everywhere dense subset (see Theorem~$1.18$ in~\cite{stone1932}). Applying to this set the so-called Gram–Schmidt orthonormalization process yields a complete orthonormal subset of $\mathfrak{G}$ and --- since $\mathfrak{G}$ is dense --- a $J$-invariant basis in $\mathscr{K}$. In the $p$-adic setting, what prevents us to repeat a similar general construction is the lack of a $p$-adic counterpart of the Gram–Schmidt process, which precludes the possibility to turn a normal set into an orthonormal basis in the invariant subspace --- thus, an invariant orthonormal basis in the Hilbert space.
\end{remark}
\begin{remark}\label{rem.2.24}
    Recall that, since $\Qe$ is separable (and spherically complete \cite{robert2013course}), every $p$-adic Hilbert space $\Hi$ is of countable-type --- equivalently, separable. In fact, any subspace of $\Hi$ itself is of countable-type (cf.\ Corollary~$2.3.14$ in~\cite{perez2010locally}), and, thus, admits a $t$-orthogonal basis for some $t\in (0,1]$ (see Theorem~$2.3.7$ in~\cite{perez2010locally}). Owing on these observations, the conclusions of Proposition~\ref{prop.2.22} are as expected. Specifically, if $Z$ is an involutive anti-unitary operator in $\Hi$, and $\invset$ a $Z$-invariant dense subspace, since $\invset$ is of countable-type, it always admits a $t$-orthogonal basis, for some $t\in (0,1]$. Then, since $\invset$ is dense in $\Hi$, this basis extends to a $Z$-invariant, $t$-orthogonal basis in the whole Hilbert space $\Hi$. Nevertheless, there is no guarantee that an \emph{orthonormal basis} exists in $\invset$ --- indeed, such a basis exists if and only if $\invset$ is an Hilbert subspace of $\Hi$~\cite{aniello2024quantum} --- and, a fortiori, that a $Z$-invariant orthonormal basis can be constructed in the Hilbert space $\Hi$.
 \end{remark}

\subsection{Construction of the isomorphism}
We now proceed to construct an isomorphism between the $p$-adic tensor product Hilbert space $\Hi\otimes\Ki$ and the space $\T(\Hi,\Ki)$ of the trace/Hilbert-Schmidt class operators. This is not simply an \scom additional result' but, rather, a consistency check that our previous construction of the tensor product of two $p$-adic Hilbert spaces --- as presented in Section~\ref{sec.3} --- provides indeed the correct counterpart to the usual tensor product of the complex setting.

Thus, let $\Phi=\{\phi_i\}_{i\in \N}$ and $\Psi=\{\psi_j\}_{j\in \N}$ be two orthonormal bases in the $p$-adic Hilbert spaces $\Hi$ and $\Ki$, respectively. Recall that $\{{^{ji}}E^{\Phi,\Psi}\}_{j,i \in \mathbb{N}}=\{\ket{\psi_j}\bra{\phi_i}\}_{j,i\in \N}\equiv \{\braket{\phi_i}{\cdot}\psi_j\}_{j,i\in \N}$ is an orthonormal basis in $\T(\Hi,\Ki)$ (see Theorem \ref{th.tracla}); hence, any $T\in \T(\Hi,\Ki)$ can be written, in a unique way, as
\begin{equation}
    T=\sum_{j,i\in \N}T_{ji}\func{\phi_i}\psi_j,\quad \mbox{where} \quad T_{ji}=\hsbraket{{^{ji}}E^{\Phi,\Psi}}{T}=\tr(\func{\phi_i}\psi_jT).
\end{equation}
\begin{proposition}\label{I(T)}
Let $\Hi$ and $\Ki$ be two $p$-adic Hilbert spaces,
    let $J_{\Phi}$ be the conjugation operator associated with the orthonormal basis $\Phi \equiv \{\phi_i\}_{i\in \N}\subset \Hi$ (cf.\ Definition \ref{def.2.13}), and let $\Psi\equiv \{\psi_j\}_{j\in \N}$ be an orthonormal basis in $\Ki$. Define the map
\begin{equation}
     \mathcal{I}\colon\T(\Hi,\Ki)\ni T\mapsto \mathcal{I}(T)\coloneqq\sum_{i,j\in \N} \braket{\psi_j}{T\phi_i}(\phi_i\otimes \psi_j)\in \Hi\otimes \Ki.
\end{equation}
    Then, we have that 
    \begin{equation}\label{eq.117n}
        \mathcal{I}(\func{v}w)=(J_{\Phi}v)\otimes w
    \end{equation}
holds for every vector     
    $v=\sum_{i\in \N}\braket{\phi_i}{v}\phi_i$ in  $\Hi$, and $w=\sum_{j\in \N}\braket{\psi_j}{v}\psi_j$ in $\Ki$, such that $\func{v}w\in \T(\Hi,\Ki)$.
\end{proposition}

\begin{proof}
    The following chain of equations holds
    \begin{align}
        \mathcal{I}(\func{v}w)&=
        \sum_{i,j\in \N} \braket{\psi_j}{w}\braket{v}{\phi_i}(\phi_i\otimes \psi_j)\nonumber \\
        &= \sum_{i,j\in \N} \braket{v}{\phi_i}\braket{\psi_j}{w}(\phi_i\otimes \psi_j)\nonumber \\
        &= \bigg (\sum_{i\in \N} \braket{v}{\phi_i}\phi_i\bigg) \otimes \bigg (\sum_{j\in \N} \braket{\psi_j}{w}\psi_j\bigg)= (J_{\Phi}v) \otimes w.
    \end{align}
\end{proof}
Next result provides the desired isomorphism between the tensor product of $p$-adic Hilbert spaces $\Hi\otimes \Ki$, and $\T(\Hi,\Ki)$.

\begin{theorem}\label{th.4.14n}
    The application $\mathcal{I}$ defined in Proposition \ref{I(T)} is an inner-product-preserving linear isomorphism between the $p$-adic Hilbert spaces $\Hi\otimes \Ki$ and $\T(\Hi,\Ki)$.  Moreover, the map 
    \begin{equation}\label{I^*}
        \mathcal{I^*}\colon\Hi\otimes \Ki \ni u\mapsto\mathcal{I^*}(u)\coloneqq\sum_{i,j\in \N} \braket{\phi_i\otimes \psi_j}{u}\func{\phi_i}\psi_j\in \T(\Hi,\Ki)
    \end{equation}
    is such that $\mathcal{I}^*=\mathcal{I}^{-1}.$
\end{theorem}

\begin{proof}
    Since $\mathcal{I}$ is linear and maps the norm-orthogonal basis $\{\func{\phi_i}\psi_j\}_{j,i\in \N}$ to $\{(J_{\Phi}\phi_i)\otimes \psi_j\}_{i,j\in \N},$ then it is IP-preserving and  injective. To prove it is also surjective, let us take $u\in \Hi\otimes \Ki$, and let $T\in \T(\Hi,\Ki)$ be written as 
    \begin{equation}
        T=\sum_{i,j\in \N} \braket{\phi_j\otimes \psi_i}{u}\func{\phi_j}\psi_i.
    \end{equation}
    Then, we have that
    \begin{equation}\label{I(T)=u}
        \begin{split}
            \mathcal{I}(T)&\coloneqq  \sum_{i,j\in \N} \braket{\psi_j}{T\phi_i}(\phi_i\otimes \psi_j)= \sum_{i,j\in \N} \bigg\langle \psi_j,\sum_{s,t\in \N}\braket{\phi_s\otimes \psi_t}{u} \psi_t\bigg\rangle\braket{\phi_s}{\phi_i}(\phi_i\otimes \psi_j)\\
            & \: =\sum_{s,t\in \N} \braket{\phi_s\otimes \psi_t}{u}\sum_{i,j\in \N} \delta_{jt}\delta_{si}(\phi_i\otimes \psi_j)= \sum_{i,j\in \N} \braket{\phi_i\otimes \psi_j}{u}(\phi_i\otimes \psi_j)=u.
        \end{split}
    \end{equation}
    By Corollary~$6.13$ in~\cite{aniello2023trace}, for an operator $T\in \mathcal{T}(\Hi,\Ki)$ we have $\lim_{i +j}\braket{\psi_j}{T\phi_i}=0.$ This, together with Remark~$6.5$ in~\cite{aniello2023trace}, entails that the double series 
    \begin{equation}
        \sum_{i\in \N}\sum_{j\in \N} \braket{\psi_j}{T\phi_i}=\sum_{j\in \N}\sum_{i\in \N} \braket{\psi_j}{T\phi_i}=\sum_{i,j\in \N} \braket{\psi_j}{T\phi_i}
    \end{equation}
    converges, and we can exchange the order of the sums in the second equality of~\eqref{I(T)=u}. Since~\eqref{I(T)=u} holds for all $u\in \Hi\otimes \Ki$, the application $\mathcal{I}$ is surjective. Moreover, by observing that
    \begin{equation}
        \hsbraket{\mathcal{I}^*\mathcal{I}(T)}{S}=\hsbraket{\mathcal{I}(T)}{\mathcal{I}(S)}=\hsbraket{T}{S},
    \end{equation}
     we can argue that  $\mathcal{I}^\ast\mathcal{I}=\Id$ and  $\mathcal{I}\mathcal{I}^\ast=\Id$, i.e., $\mathcal{I}^\ast=\mathcal{I}^{-1}$.
    
    To prove~\eqref{I^*}, it suffices to note that a straightforward computation yields 
    \begin{equation}
        \bigg\langle T, \sum_{i,j\in\N} \braket{\phi_i\otimes \psi_j}{u}\func{\phi_i}\psi_j\bigg\rangle_{\mbox{\tiny $\T(\Hi,\Ki)$}}=\braket{\mathcal{I}(T)}{u}=\braket{T}{\mathcal{I}^*(u)}.
    \end{equation}
    for every $T\in \T(\Hi,\Ki)$.
\end{proof}
To conclude this section we collect, in the following result, some useful consequences of Theorem~\ref{th.4.14n}
\begin{corollary}
Let $\Hi$ and $\Ki$ be two $p$-adic Hilbert spaces, and let $J_\Phi$ be the conjugation operator associated with the orthonormal basis $\Phi\equiv\{\phi_n\}_{n\in\N}\subset\Hi$. Then, for the operator $\func{v}w\in \T(\Hi,\Ki)$ --- where $v\in \Hi$ and $w\in \Ki$ are any two vectors --- the following equalities hold
\begin{enumerate}[label=(\roman*)]
\item \label{item.i2.cor.4.15} $\func{v}w=\mathcal{I}^*((J_{\Phi}v)\otimes w)$;
\item \label{item.i3.cor.4.15} $\mathcal{I}^*(v \otimes w)=\func{J_{\Phi}v}w$.
\end{enumerate}
\end{corollary}
\begin{proof}
For~\ref{item.i2.cor.4.15}, it suffices to observe  that, by the second assertion in Theorem~\ref{th.4.14n}, $\mathcal{I}^\ast=\mathcal{I}^{-1}$. Hence, applying $\mathcal{I}^\ast$ to both sides of~\eqref{eq.117n}, we have:
\begin{align}
\mathcal{I}^\ast\mathcal{I}(\func{v}w)&=\mathcal{I}^\ast\big((J_\Phi v)\otimes w\big)\nonumber\\
\func{v}w&=\mathcal{I}^\ast\big((J_\Phi v)\otimes w\big).
\end{align}
For~\ref{item.i3.cor.4.15}, exploiting the identity~\eqref{eq.117n} in Proposition~\ref{I(T)}, and recalling that $J_\Phi$ is an involutive operator --- i.e., $J_\Phi^2=\mathrm{Id}$  --- we can argue that
\begin{align}
    \func{J_\Phi v}w&=\mathcal{I}^\ast\big(J_\Phi(J_\Phi v)\otimes w\big)\nonumber\\
    &=\mathcal{I}^\ast\big((J_\Phi^2 v)\otimes w\big)\nonumber\\
    &=\mathcal{I}^\ast(v\otimes w).
\end{align}
\end{proof}

\begin{remark}
The construction of the isomorphism provided in Theorem~\ref{th.4.14n} --- between the tensor product Hilbert space $\Hi\otimes\Ki$ and $\T(\Hi,\Ki)$ --- rests, in a crucial way, on the correspondence between the compact operator $\func{v}w\in\T(\Hi,\Ki)$, $v\in\Hi$, $w\in\Ki$, and the simple tensor $(J_\Phi v )\otimes w$. It is worth emphasizing that this isomorphism involves the conjugation operator $J_\Phi$ --- associated with the orthonormal basis $\Phi$ in $\Hi$ --- thereby mirroring the canonical antilinear isomorphism between the tensor product of complex Hilbert spaces, and the space of the Hilbert-Schmidt operators.
\end{remark}


\section{Tensor product of subspaces}
\label{sec.5}

In this section, we address the tensor product structure of subspaces in a (tensor product) $p$-adic Hilbert space. Once again, the well known results from the complex framework~\cite{ryan2002introduction} will provide us a useful road map, but, when switching to the $p$-adic setting, the emergence of non-trivial and distinctive peculiarities are encountered. 

In what follows, for the sake of notational convenience, we set $c_0(I,\Qe)\equiv c_0(I)$ --- where $I$ is a finite set or $I=\N$. We start by proving the following
\begin{proposition}\label{prop.5.1}
Let $(X,\No{\cdot})$ be a $p$-adic Banach space. Then, there exists a canonical $p$-adic Banach space isomorphism between $c_0(I)\ptpc X$ and $c_0(I,X)$.
\end{proposition}
\begin{proof}
To construct the isomorphism, we proceed in two steps: first, we prove that there exists a natural isomorphism between $c_0(I)\ptp X$ and $c_0(I,X)$; then, we show that this map admits a unique extension to the complete tensor product space $c_0(I)\ptpc X$. Specifically, let us consider the map $\Ji_X\colon c_0(I)\ptp X\rightarrow c_0(I,X)$ defined by
\begin{equation}
c_0(I)\ptp X\ni \lambda\hot x\mapsto
\Ji_X(\lambda\hot x)\coloneqq\{\lambda_i x\}_{i\in I}\in c_0(I,X),
\end{equation}
where $\lambda=\{\lambda_i\}_{i\in I}\in c_0(I)$, and $x\in X$. Since $\lambda\in c_0(I)$ it is clear that $\{\lambda_i x\}_{i\in I}$ is an element of $c_0(I,X)$ for every $x$ in $X$, i.e.,
\begin{equation}
\lim_i \lambda_i x=0,\quad
\No{\{\lambda_i x\}_{i\in I}}_{\infty}=\sup_i\abs{\lambda_i}\,\No{x}=\No{\lambda}_{\infty}\,\No{x}.
\end{equation}
We now aim to show that $\Ji_X$ defines a norm-preserving map between $c_0(I)\ptp X$ and $c_0(I,X)$. Indeed, let 
 $u=\sum_{j=1}^m \lambda_j\hot x_j$ be a possible representation of $u\in c_0(I)\ptp X$ where, without loss of generality, we can assume $\{x_j\}_{j=1}^m$ to be a norm-orthogonal system in $X$ (cf.\ Lemma~$10.2.6$ in~\cite{perez2010locally}).
Taking into account that $\lambda_j=\{\lambda_{i,j}\}_{i\in I}$ for $j=1,\ldots,m$, we have:
\begin{equation}\label{eq.94}
    \Ji_X(u)\equiv\{u_i\}_{i\in I}=\bigg{\{}\sum_{j=1}^m \lambda_{i,j}x_j\bigg{\}}_{i\in I}.
\end{equation}
Therefore, it follows that 
\begin{align}\label{eq.95}
\No{\Ji_X(u)}_{\infty}=\bigg{\|}{\bigg{\{} \sum_{j=1}^m \lambda_{
    i,j}x_j\bigg{\}}}_{i\in I}\bigg{\|}_{\infty}=
    \sup_{i\in I}\bigg\|\sum_{j= 1}^m\lambda_{i,j}x_j\bigg\|&\leq \sup_{i\in I}\max_{1\leq j\leq m}\abs{\lambda_{i,j}}\,\No{x_j}\nonumber\\
    &=\max_{1\leq j\leq m}\No{\lambda_j}_{\infty}\,\No{x_j}=\No{u}_\pi.
\end{align}
To prove the reverse inequality, let us fix a representation $\sum_{j=1}^m\lambda_j\hot x_j$ of $u$ in $c_0(I)\ptp X$ --- where, again, we assume $\{x_j\}_{j=1}^m$ to be a norm-orthogonal system in $X$.
Let us consider the series $\sum_{n\in I}e_n\hot u_n$ where $\{u_n\}_{n\in I}=\{\sum_{j=1}^m\lambda_{n,j}x_j\}_{n\in I}=\Ji(u)$ and  $\boldsymbol{e}=\{e_n\}_{n\in I}$ is the so-called canonical basis in $c_0(I)$ (cf.\ Example~\ref{exa.2.4}). Let $\Pi_k\colon c_0(I)\rightarrow c_0(K)$ be the map --- where $K=\{1,\ldots,k\}$, $\card{K}\leq\card{I}$ --- denoting the projection of the elements in $c_0(I)$ onto the first $k$ coordinates. We then have that
\begin{equation}
    \Pi_k(\lambda)=\sum_{n=1}^k \lambda_n e_n,\quad \forall \lambda=\{\lambda_n\}_{n\in I}\in c_0(I).
\end{equation}
It is clear that $\Pi_k(\lambda)\rightarrow \lambda$ for $k\rightarrow \infty$ (if $I$ is a finite set, for $k$ equal to $\text{card}(I)$) and therefore
\begin{align}
    \bigg\|u-\sum_{n=1}^k e_n\hot u_n \bigg\|_\pi&=\bigg\|\sum_{j=1}^m\lambda_j\hot x_j -\sum_{n=1}^k\sum_{j=1}^m e_n \hot \lambda_{n,j} x_j\bigg\|_\pi\nonumber\\
    &=\bigg\|\sum_{j=1}^m\bigg(\lambda_j\hot x_j-\sum_{n=1}^k \lambda_{n,j}e_n\hot x_j\bigg)\bigg\|_\pi\nonumber\\
    &=\bigg\|\sum_{j=1}^m\big(\lambda_j-\Pi_k(\lambda_j) \big)\hot x_j\bigg\|_\pi\nonumber\\
    &\leq \max_{1\leq j\leq m}\No{\lambda_j-\Pi_k(\lambda_j)}_{\infty}\,\No{x_j}.
\end{align}
Thus, for $k\rightarrow \infty$ (or, if $I$ is a finite set, for $k$ equal to $\text{card}(I)$), $\No{u-\sum_{n=1}^k e_n\hot u_n}_\pi=0$ and hence we proved that $\sum_{n=1}^k e_n\hot u_n$ converges to $u$ in the projective norm $\No{\cdot}_\pi$. Now, we have:
\begin{equation}\label{eq.98}
\No{u}_\pi=\bigg\|\sum_{i\in I} e_i\hot u_i\bigg\|_\pi\leq \sup_{i}\No{u_i}\,\No{e_i}_\infty=\sup_i\No{u_i}=\No{\Ji_X(u)}_{\infty}.
\end{equation}
Therefore, taking into account~\eqref{eq.95} and~\eqref{eq.98}, we conclude that $\No{\Ji_X(u)}_{\infty}=\No{u}_\pi$, i.e., $\Ji_X$ is a norm-preserving map (an isometry) between the $p$-adic normed spaces $c_0(I)\ptp X$ and $c_0(I,X)$. On the other hand, since $c_0(I,X)$ is complete, $\Ji_X$ admits a unique bounded (Hahn-Banach) extension to an isometric operator  --- which we still denote with the same symbol --- from the completed projective tensor product $c_0(I)\ptpc X$ into $c_0(I,X)$. 

To prove that $\Ji_X$ is an isomorphism of $p$-adic Banach spaces, we are left to show that it is a surjective map. To this end, let $x=\{x_i\}_{i\in I}$ be a sequence in $c_0(I,X)$. If we consider the series $\sum_{i\in I} e_i\hot x_i$, it is clear that its image under $\Ji_X$ coincides with $x=\{x_i\}_{i\in I}$. Hence, we have only to show that $\sum_{i\in I} e_i\hot x_i$ converges in $c_0(I)\ptpc X$. If $\card{I}<\infty$, this is certainly true. On the other hand, if $I= \N$, since $c_0(I)\ptpc X$ is complete, it suffices to show that $\big\{\sum_{i=1}^n e_i\hot x_i\big\}_{n\in\N}$ is a $p$-adic Cauchy sequence. This simply follows by observing that
\begin{equation}
    \bigg\|\sum_{i=1}^{n+1}e_i\hot x_i-\sum_{i=1}^n e_i\hot x_i \bigg\|_\pi=\No{e_{n+1}\hot x_{n+1}}_\pi=\No{e_{n+1}}_\infty\,\No{x_{n+1}},
\end{equation}
and recalling that, since $\{x_n\}_{n\in I}\in c_0(I,X)$, $\lim_n x_n=0$. Hence, $\Ji_X$ is a surjective isometry --- i.e., an isometric isomorphism --- between the $p$-adic Banach spaces $c_0(I)\ptpc X$ and $c_0(I,X)$.
\end{proof}
Proposition~\ref{prop.5.1} allows us to identify the tensor product $c_0(I)\ptpc X$ with the $p$-adic Banach space $c_0(I,X)$ of zero-convergent sequences in $X$. As a primary implication of this result, we now prove the following  
\begin{corollary}\label{prop.5.2}
Let $(Y,\No{\cdot})$ be a $p$-adic Banach space, and let $X\subset Y$ be a closed subspace of $Y$. Then, $c_0(I)\ptpc X$ is a subspace of $c_0(I)\ptpc Y$.
\end{corollary}
\begin{proof}
From Proposition~\ref{prop.5.1}, we know that the map $\Ji_Y\colon c_0(I)\ptpc Y\rightarrow c_0(I,Y)$ (respectively $\Ji_X\colon c_0(I)\ptpc X\rightarrow c_0(I,X)$) provides an isometric isomorphism of $p$-adic Banach spaces; i.e., 
\begin{equation}
    c_0(I)\ptpc Y\cong c_0(I,Y),\quad c_0(I)\ptpc X\cong c_0(I,X).
\end{equation}
On the other hand, $c_0(I,X)$ is a subspace of $c_0(I,Y)$ and, hence, $c_0(I)\ptpc X$ is a subspace of $c_0(I)\ptpc Y$, that is what we wanted to prove.
\end{proof}
\begin{remark}\label{rem.5.5}
 Proposition~\ref{prop.5.2} is a $p$-adic counterpart of a well known result for the projective tensor product of complex Banach spaces. Specifically, let $\X$ and $\Y$ be two (complex) Banach spaces, and let $\W$ be a closed subspace of $\Y$. Denoting by $\X\ptpc\W$ the completion of the projective tensor product of $\X$ and $\W$, it is readily seen that $\X\ptpc\W$ will not be, in general, a subspace of $\X\ptpc\Y$, unless  $\X=\ell_1(I)$ is the Banach space of absolutely convergent complex sequences indexed by $I$ --- in this case, one can indeed show that $\ell_1(I)\ptpc \W$ is a subspace of $\ell_1(I)\ptpc\Y$; see Proposition $2.7$ in~\cite{ryan2002introduction}. Therefore, the result of Proposition~\ref{prop.5.2} is as expected once observed that, in the $p$-adic setting, $\ell_1(I)\equiv c_0(I)$~\cite{natarajan2019sequence}. 
\end{remark}
Our aim is now to investigate how the tensor product operation behaves w.r.t.\ the $p$-adic Hilbert subspace structure. 
We start by recalling that, in the $p$-adic setting, there are two natural notions of subspaces that can be defined for a $p$-adic Hilbert space (see~\cite{aniello2024quantum} for a detailed discussion); specifically, we have:
\begin{definition}
    Let $\Hi$ be a $p$-adic Hilbert space, and let $\Ws\subset \Hi$ be a norm-closed subspace of $\Hi$.
    We say that $\Ws$ is a \emph{Hilbert subspace} of $\Hi$ if it is a $p$-adic Hilbert space itself --- equivalently, if $\Ws$ admits an orthonormal basis $\Psi\equiv\{\psi_i\}_{i\in I}\subset \Ws$, $\card{I}=\dim \Ws$. We say that $\Ws$ is a \emph{regular subspace} of $\Hi$, if it is a Hilbert subspace, and if some (therefore, any) orthonormal basis of $\Ws$ can be extended to an orthonormal basis of $\Hi$. 
\end{definition}
Clearly, a regular subspace is a Hilbert subspace, while the contrary is not generally true.
\begin{remark}
Let $\Ws$ be a norm-closed subspace of a $p$-adic Hilbert space $\Hi$. Denoting by $\Ws^\perp$ the (closed) subspace defined as
\begin{equation}
\Ws^\perp\coloneqq \{x\in\Hi\mid \braket{x}{w}=0,\;\; \forall w\in\Ws\},
\end{equation}
it is possible to prove that if $\Ws$ is a regular subspace, then $\Ws^\perp$ is a regular subspace too. If $\Ws^\perp$ is a Hilbert subspace of $\Hi$, and if the relation
\begin{equation}
    \No{w+v}=\max\{\No{w},\No{v}\}
\end{equation}
holds for every $v\in\Ws$ and $w\in\Ws^\perp$,
then $\Ws$ and $\Ws^\perp$ are regular subspaces of $\Hi$~\cite{aniello2024quantum}.
\end{remark}
\begin{proposition}\label{cor.5.2}
Let $\Hi$ and $\Ki$ be two $p$-adic Hilbert spaces, and let $\Psi\equiv\{\psi_i\}_{i\in I}$ be an orthonormal basis in $\Hi$. Then, the map defined as 
\begin{equation}
\teniso_\Psi\colon \Hi\otimes \Ki\ni u\mapsto \teniso_\Psi(u)\coloneqq\sum_{j=1}^m W_\Psi(x_j)\hot y_j\in\CH(I)\otimes \Ki
\end{equation}
--- where $W_\Psi\colon \Hi\rightarrow \CH(I)$ is the canonical isomorphism of $\Hi$ into the coordinate Hilbert space $\CH(I)$, and $u=\sum_{j=1}^m x_j\hot y_j$, $x_j\in\Hi$, $y_j\in\Ki$, any representation of $u$ --- is an isomorphism between the $p$-adic Hilbert spaces $\Hi\otimes \Ki$ and $\CH(I)\otimes \Ki$.
\end{proposition}
\begin{proof}
It is clear that $\teniso_\Psi$ is a surjective map. Moreover, by observing that
\begin{align}
\No{\teniso_\Psi(u)}_\pi&=\bigg\|\sum_{j=1}^m W_\Psi(x_j)\hot y_j\bigg\|\notag\\
&=\inf\bigg\{\max_{1\leq j\leq m}\No{W_\Psi(x_j)}\,\No{y_j}\bmid u=\sum_{j=1}^m x_j\hot y_j,\,\,x_j\in\Hi,y_j\in\Ki\bigg\}\nonumber\\
&=\inf\bigg\{\max_{1\leq j\leq m} \No{x_j}\,\No{y_j}\bmid u=\sum_{j=1}^m x_j\hot y_j\,\,x_j\in\Hi,y_j\in\Ki\bigg\}\nonumber\\
&=\No{u}_\pi,
\end{align}
we can argue that $\teniso_\Psi$ is  norm-preserving. Next, let $u=\sum_{j=1}^m x_j\hot y_j$ and $w=\sum_{k=1}^n f_k\hot g_k$, where $x_j,f_k\in\Hi$, $y_j,g_k\in\Ki$, be two elements of $\Hi\ptpc \Ki$. Then, we have:
\begin{align}\label{eq.149}
    \braket{\teniso_\Psi(u)}{\teniso_\Psi(w)}_\pi&=\Big\langle\sum_{j=1}^m W_\Psi(x_j)\hot y_j,\sum_{k=1}^n W_\Psi(f_k)\hot g_k\Big\rangle_\pi\nonumber\\
    &=\sum_{j=1}^m\sum_{k=1}^n\big\langle W_\Psi(x_j),W_\Psi(f_k)\big\rangle_\Hi\braket{y_j}{g_k}_\Ki\nonumber\\
    &=\sum_{j=1}^m\sum_{k=1}^n \braket{\breve{x}_j}{\breve{f}_k}_{\CH(I)}\braket{y_j}{g_k}_\Ki.
\end{align}
where we denoted by $\inprd_\Hi$, $\inprd_\Ki$, and $\inprd_{\CH(I)}$ the $p$-adic inner product defined on $\Hi$, $\Ki$, and $\CH(I)$ respectively, and where we set $\breve{x}_j\equiv\{\braket{\psi_i}{x_{j,i}}\}_{i\in I}$, $\breve{f}_k\equiv\{\braket{\psi_i}{f_{k,i}}\}_{i\in I}$ (see Example~\ref{exa.2.13}).
Observing that the last line of~\eqref{eq.149} is the natural inner product over $\CH(I)\ptpc \Ki$, we proved that $\teniso_\Psi$ is an inner-product preserving map too. Concluding, $\teniso_\Psi$ is a surjective, norm-preserving, inner-product-preserving map, i.e., it is an isomorphism of $p$-adic Hilbert spaces.
\end{proof}
We now come to the main result of this section
\begin{proposition}\label{prop.5.5}
Let $\Hi$, $\Ki$ be two $p$-adic Hilbert spaces, and let $\Ws$ be a Hilbert subspace of $\Ki$. Then $\Hi\otimes\Ws$ is a Hilbert subspace of $\Hi\otimes \Ki$. Moreover, if $\Ws$ is regular, then $\Hi\otimes\Ws$ is a regular subspace of $\Hi\otimes \Ki$.
\end{proposition}
\begin{proof}
Let $\Psi\equiv\{\psi_i\}_{i\in I}$ be an orthonormal basis in $\Hi$ ($\dim\Hi=\card{I}$). The map  $\teniso_\Psi$ in Proposition~\ref{cor.5.2} provides an isomorphism between the $p$-adic Hilbert spaces $\Hi\otimes\Ki$ and $\CH(I)\otimes\Ki$. In particular, this also entails that $\Hi\ptpc\Ki\cong c_0(I)\ptpc\Ki$ are isomorphic $p$-adic Banach spaces (recall that, we denoted with $c_0(I)$ the carrier Banach space of $\CH(I)$, see Example~\ref{exa.2.4}). In a similar way, if $\Ws$ is a Hilbert subspace of $\Ki$, we have the ($p$-adic Banach space isomorphism)  $\Hi\ptpc\Ws\cong c_0(I)\ptpc \Ws$. By Proposition~\ref{prop.5.1}, we know that $c_0(I)\ptpc \Ws$ is a subspace of $c_0(I)\ptpc\Ki$ --- and, therefore, that $\CH(I)\otimes\Ws$ is a subspace of $\CH(I)\otimes\Ki$, as well. Then, we have the following inclusion
\begin{equation}
    \Hi\otimes\Ws\cong \CH(I)\otimes\Ws\subset\CH(I)\otimes\Ki\cong\Hi\otimes\Ki,
\end{equation}
and since $\Hi\otimes \Ws$ is a Hilbert space itself --- and norm-closed in $\Hi\otimes \Ki$ --- it is a Hilbert subspace of $\Hi\otimes\Ki$.  
Assume now $\Ws\subset\Ki$ be a regular subspace. This entails that there exists an orthonormal basis $\Phi_K\equiv\{\phi_k\}_{k\in K}$ of $\Ki$ such that, for some subset $J$ of $K$ --- putting $\Phi_J\equiv\{\phi_j\}_{j\in J}$ --- $\Phi_J$ is an orthonormal basis of $\Ws$. Clearly, $\Ws$ is a Hilbert subspace of $\Ki$; thus, $\Hi\otimes \Ws$ is a Hilbert subspace of $\Hi\otimes\Ki$. Moreover, from Theorem~\ref{th.3.9}, we know that the set $\{\psi_i\hot \phi_j\}_{i\in I, j\in J}$ --- where $\Psi_I\equiv\{\psi_i\}_{i\in I}$ is an orthonormal basis in $\Hi$ --- provides an orthonormal basis in $\Hi\otimes \Ws$. Then, the set $\{\psi_i\hot\phi_k\}_{i\in I, k\in K}$ --- obtained by extending the basis $\Phi_J\subset\Ws$ to the basis $\Phi_K$ of $\Ki$ --- is, in a natural way, an orthonormal basis of $\Hi\otimes\Ki$. Concluding, we have constructed an orthonormal basis in $\Hi\otimes\Ws$ which can be extended to an orthonormal basis of $\Hi\otimes\Ki$, i.e., $\Hi\otimes\Ws$ is a regular Hilbert subspace of $\Hi\otimes\Ki$.
\end{proof}


\section{Conclusions}
\label{sec.6}
This paper is a natural continuation of our previous investigations on 
$p$-adic quantum mechanics~\cite{aniello2023states,aniello2023trace,aniello2024quantum,aniello2024}. In particular, in~\cite{aniello2023states,aniello2023trace,aniello2024quantum}, we introduced a model of quantum mechanics where the carrier vector space of physical states is an abstract Hilbert space over a quadratic extension $\Qe$ of $\Q_p$. When dealing with composite systems, it is natural to conjecture that --- analogously to what happens in the usual (complex) formulation of quantum mechanics ---
the fundamental mathematical framework should involve the tensor product of two Hilbert spaces. It is precisely this issue that motivates us to introduce, in this work, a notion of tensor product tailored to our specific $p$-adic setting.

Let us briefly summarize our main achievements: 
\begin{itemize}
    \item We have first defined the algebraic tensor product $\Hi\hot_\alpha \Ki$ of two $p$-adic Hilbert spaces --- which involves solely the linear structure of these spaces. Next, we have endowed $\Hi\hot_\alpha \Ki$ with a suitable non-Archimedean norm, thereby achieving a $p$-adic normed space $\Hi\ptp\Ki$ out of the initial algebraic tensor product. Here, we have noted that this norm may be regarded as the ultrametric analogue of the standard \emph{projective norm} associated with the algebraic tensor product of two real or complex Banach spaces --- and, accordingly, we denoted the associated Banach space completion by $\Hi\ptpc\Ki$. As a last step of our construction, we have endowed $\Hi\ptpc\Ki$ with the inner product $\inprd_\pi$, and formed a new orthonormal basis $\Xi\equiv\{\xi_k\}_{k\in K}\equiv\{\phi_i\otimes \psi_j\}_{i\in I,\, j\in J}$ out of the orthonormal bases $\Phi\equiv\{\phi_n\}_{n\in\N}$ and $\Psi\equiv\{\psi_m\}_{m\in\N}$ in the initial Hilbert spaces $\Hi$ and $\Ki$. Eventually, this has led us to the ($p$-adic Hilbert space) tensor product $\Hi\otimes\Ki\equiv\big(\Hi\ptpc \Ki, \PNo{\cdot}, \inprd_\pi,\bm{\Xi}\equiv\{\xi_k\}_{k\in K}\equiv\{\phi_i\otimes \psi_j\}_{i\in I,\, j\in J}\big)$. 
    \item As a next step, we have characterized the set of anti-unitary operators $\AU$ in a $p$-adic Hilbert space $\Hi$. To this end, we have first introduced a suitable notion of a conjugation operator $J_\Phi$ --- associated with any orthonormal basis $\Phi\equiv\{\phi_n\}_{n\in\N}$ in $\Hi$ --- defined by requiring that $J_\Phi$ is precisely that conjugate-linear map which leaves invariant the subset $\Phi\subset\Hi$, i.e., such that $J_\Phi \phi_n=\phi_n$, $\forall n\in\N$. An anti-unitary operator has then been defined as the composition of a unitary operator in $\Hi$, and the conjugation operator $J_\Phi$. After having characterized the set $\AU$, we have highlighted a key difference between $p$-adic and complex conjugation operators. In fact, whereas, in the complex setting, a conjugation operator in a Hilbert space $\mathscr{H}$ can be (equivalently) defined as an involutive anti-unitary operator $Z\colon \mathscr{H}\rightarrow \mathscr{H}$ --- in which case one can prove there always exits an orthonormal basis in $\mathscr{H}$ preserved by the action of $Z$ --- in the $p$-adic setting, it is not always possible to construct an invariant orthonormal basis $\Phi$ in $\Hi$, such that $Z\equiv J_\Phi$ is the associated conjugation operator.
    \item As a \scom consistency check' of our previous construction of the $p$-adic tensor product Hilbert space $\Hi\otimes \Ki$, we have exploited the conjugation operator $J_\Phi$ to define a conjugate-linear isomorphism between $\Hi \otimes \Ki$, and the $p$-adic Hilbert space $\T(\Hi,\Ki)$ of trace/Hilbert-Schmidt class operators from $\Hi$ to $\Ki$. This establishes a clear analogy with the familiar tensor product structure of complex Hilbert spaces.
    \item Lastly, we have studied how the tensor product $\Hi\otimes\Ki$ behaves w.r.t.\ the subspaces of the original $p$-adic Hilbert spaces $\Hi$ and $\Ki$. Here, we have first proved that, similarly to the complex setting --- where one can show there exists a canonical isomorphism between a complex or real tensor product Banach space $\mathscr{X}\ptpc \ell_1$ and the sequence space $\ell_1(\mathscr{X})$ --- given a $p$-adic Banach space $(X,\No{\cdot})$, it is possible to construct the isomorphism $c_0(I)\ptpc X\cong c_0(I,X)$. However, unlike the complex framework, we have proved here a much stronger result: We showed that the $p$-adic tensor product operation preservers the subspace structure of a $p$-adic Hilbert space. Namely, that if $\Hi$ and $\Ki$ are two $p$-adic Hilbert spaces, and if $\mathcal{W}\subset \Ki$ is an Hilbert (or a regular) subspace of $\Ki$, then $\Hi\otimes \mathcal{W}$ is an Hilbert (resp.\ regular) subspace of $\Hi\otimes \Ki$.
\end{itemize}
Let us now provide a \dcom bird’s-eye view" of some possible continuations of this work. Indeed, a first interesting question one can address at this point, is what \emph{reflexivity} and \emph{dual properties} characterize the tensor product $\Hi \otimes \Ki$. In fact, we already know that the dual of a $p$-adic Hilbert space $\Hi$ is not a $p$-adic Hilbert space itself --- if fact, it is a $p$-adic Banach space --- and we observed that $\Hi$ is, in general, not reflexive, but only pseudo-reflexive~\cite{aniello2023trace,aniello2024quantum}. Hence, a detailed study of the dual and the bidual of $\Hi\otimes\Ki$ could unveil new peculiarities of $p$-adic Hilbert spaces not observed in the standard complex setting.

Once the $p$-adic tensor product Hilbert space $\Hi\otimes \Ki$ is constructed, a further natural step would be to investigate the tensor product structure of the set of adjointable, and trace/Hilbert-Schmidt class operators in $\Hi\otimes \Ki$. In particular, the former will provide a suitable notion of an observable, and the latter, that of a state, of a composite quantum system. This will also allow us to introduce a natural notion of \emph{entangled state} in the $p$-adic setting and, more in general, to investigate the separability properties of $p$-adic quantum systems.

\appendix

\section{\texorpdfstring{Proofs of Theorems~\ref{th.2.5},~\ref{th.2.7} and~\ref{th.tracla}}{}}\label{app.a1}
In this appendix, we present detailed proofs of the theorems stated in Subsection~\ref{subsec.2.2}. Since their arguments closely mirror those of the single Hilbert space setting (see, in particular, Theorems~$4.2$,~$4.4$ and~$6.38$ in~\cite{aniello2023trace}), and only require minor modifications, we opted to include them here, leaving the main text to the more original material. We start from
\begin{proof}[Proof of Theorem \ref{th.2.5}]
Let $\Hi$ and $\Ki$ be $p$-adic Hilbert spaces with orthonormal bases $\Phi\equiv\{\phi_n\}_{n\in \N}$, and $\Psi\equiv\{\psi_m\}_{m\in \N}$, respectively. From relations~\eqref{eq.bondop},~\eqref{eq.bondop2}, 
it is clear that $\Ba{\Hi}{\Ki}\subset (\Hi,\Ki)_{\Phi,\Psi}$. To prove the reverse inclusion, and, thus, the first equality in~\eqref{eq.20}, it suffices to observe that if $B\equiv\op_{\Phi,\Psi}(B_{mn})\in (\Hi,\Ki)_{\Phi,\Psi}$ then, for any $\zeta$ in $\Hi$, we have:
\begin{equation}\label{eq.22}
\No{B\zeta}=\sup_{m}\bigg|\sum_nB_{mn}\braket{\phi_n}{\zeta}\bigg|\leq\sup_{m,n}\abs{B_{mn}}\,\abs{\braket{\phi_n}{\zeta}}\leq \No{\zeta}\sup_{m,n}\abs{B_{mn}},
\end{equation}
i.e., $B\equiv\op_{\Phi,\Psi}(B_{mn})\in\Ba{\Hi}{\Ki}$. The second equality in~\eqref{eq.20} simply follows by relations~\eqref{eq.bondop}.
To prove~\eqref{eq.21}, one can notice that if $B$ is a bounded operator, then
\begin{equation}\label{eq.28napp}
\No{B\phi_n}=\bigg\|\bigg(\sum_m\sum_s\braket{\psi_m}{B\phi_s}\braket{\phi_s}{\cdot}\psi_m\bigg)(\phi_n)\bigg\|=\bigg\|\sum_m\sum_s\braket{\psi_m}{B\phi_s}\delta_{sn}\psi_m\bigg\|=\sup_m\abs{B_{mn}},
\end{equation}
where, in the first equality, we used the first series expansion in~\eqref{eq.matop} of the bounded operator $B$ --- converging in the strong operator topology --- while the last quality follows from the Parseval identity (cf.~\eqref{eq.11n}). Hence, from~\eqref{eq.22}, and taking into account relation~\eqref{eq.28napp},  we can conclude that $\No{B}=\sup_{m,n}\abs{B_{mn}}=\sup_n\No{B\phi_n}$.
\end{proof}
We next prove the main result providing the characterization of adjointable operators acting between two $p$-adic Hilbert spaces $\Hi$ and $\Ki$.
\begin{proof}[Proof of Theorem \ref{th.2.7}]
Let $B$ be a bounded operator. Owing to the second series expansion in~\eqref{eq.matop} we have, for all $\zeta\in\Ki$ and $\chi\in\Hi$, that
\begin{equation}\label{eq.31napp}
    \braket{\zeta}{B\chi}=\sum_n\sum_m B_{mn}\braket{\zeta}{\psi_m}\braket{\phi_n}{\chi}=\sum_n\bigg(\overline{\sum_m\overline{B_{mn}}\braket{\psi_m}{\zeta}}\bigg)\braket{\phi_n}{\chi}.
\end{equation}
Since $\chi$ is an arbitrary vector in $\Hi$, we can argue that  
\begin{equation}\label{eq.36napp}
\{\xi^\zeta_n\equiv\sum_m\overline{B_{mn}}\braket{\psi_m}{\zeta}\}_{n\in\N}\in\ell^\infty,
\end{equation}
i.e., $\{\xi_n^\zeta\}_{n\in\N}$ 
is a bounded sequence. Using the isomorphism $\mathcal{E}_{\Phi}\colon\ell^\infty\rightarrow\Hi^\prime$ (see Remark~\ref{rem.2.4}), we have that
\begin{equation}
    \braket{\zeta}{B\chi}=\sum_n\sum_m\overline{\xi_n^\zeta}\braket{\phi_n}{\chi}=\big(\mathcal{E}_\Phi(\overline{\xi^\zeta})\big)(\chi).
\end{equation}
Assume now that $B\in\Bad{\Hi}{\Ki}$. Then, for some $\eta=B^\ast\zeta\in\Hi$, the following relation holds 
\begin{equation}
    \big(\mathcal{E}_\Phi(\overline{\xi^\zeta})\big)(\chi)=\braket{\zeta}{B\chi}=\braket{B^\ast\zeta}{\chi},\quad \forall\chi\in\Hi;
\end{equation}
from this, and exploiting Proposition~$3.36$ in~\cite{aniello2023trace}, we can argue that $\{\xi_n^\zeta\}_{n\in\N}\in c_0$ and, moreover, that
\begin{equation}
\sum_n\xi_n^\zeta\phi_n=B^\ast\zeta\in\Hi.
\end{equation}
Using the explicit form of $\xi_n^\zeta$ in~\eqref{eq.36napp}, we can retrieve the matrix representation of the operator $B^\ast$, i.e.,
\begin{equation}
B^\ast=\op_{\Psi,\Phi}(B^\ast_{mn})=\op_{\Psi,\Phi}(\overline{B_{nm}}).
\end{equation}
Concluding, the adjoint $B^\ast$ of $B$ is a bounded (see~\eqref{eq.30}), all-over matrix operator and, in the light of Theorem~\ref{th.2.5}, its matrix coefficients $B^\ast_{mn}=\overline{B_{nm}}$ have to satisfy conditions~\ref{cond.ad.i}--\ref{cond.ad.iii}. By further taking into account that $\sup_{m,n}\abs{\overline{B_{nm}}}=\sup_{m,n}\abs{B_{mn}}$, relation~\eqref{norm.adj} is also proved.

Conversely, let us suppose that $B=\op_{\Phi,\Psi}(B_{mn})$ is such that its matrix elements $B_{mn}$ satisfy conditions~\ref{cond.ad.i}-\ref{cond.ad.iii}. Then, from~\ref{cond.ad.i} and~\ref{cond.ad.ii}, it follows that $B$ is an all-over bounded matrix operator, and expansion~\eqref{eq.31napp} holds for every $\zeta$ in $\Ki$ and $\chi$ in $\Hi$. Moreover if, for every $m\in\N$, $\lim_n B_{mn}=0$ then, setting $B^\ast=\op_{\Psi,\Phi}(B^\ast_{mn})=\op_{\Psi,\Phi}(\overline{B_{nm}})$, we have that the series $\sum_m\sum_n B^\ast_{mn}\braket{\psi_n}{\zeta}\phi_m$ must converge to some vector $\eta=\op_{\Psi,\Phi}(B_{mn}^\ast)\zeta\in\Hi$. Therefore, for every $\zeta$ in $\Ki$, there is some $\eta=\op_{\Psi,\Phi}(B_{mn}^\ast)\zeta$ in $\Hi$, such that
\begin{equation}\label{eq.33app}
\braket{\zeta}{B\chi}=\braket{B^\ast\zeta}{\chi}=\braket{\eta}{\chi}.
\end{equation}
Since relation~\eqref{eq.33app} determines $B^\ast$ uniquely, this concludes the proof.
\end{proof} 
We conclude with 
\begin{proof}[Proof of Theorem~\ref{th.tracla}]
Let $\Hi$ and $\Ki$ be $p$-adic Hilbert spaces over the same quadratic extension $\Qe$ of $\Q_p$, and let $\Phi\equiv\{\phi_n\}_{n\in\N}$ and $\Psi\equiv\{\psi_m\}_{m\in\N}$ be two orthonormal bases in $\Hi$ and $\Ki$ respectively. To prove that $\big\langle \T(\Hi,\Ki),\No{\cdot},\HSprod, \{\E{jk}{\Phi,\Psi}{}\}_{j,k\in\N}\big\rangle$ is a $p$-adic Hilbert space, we first show that $\HSprod$ is a (non-degenerate)  $p$-adic inner product. Indeed, let $T\in\T(\Hi,\Ki)$; we have:
\begin{equation}
\hsbraket{\E{jk}{\Phi,\Psi}{}}{T}=\tr\big(\ketbra{\phi_k}{\psi_j}T\big)=T_{jk},
\end{equation}
where, in the first equality, we used the fact that $\big(\E{jk}{\Phi,\Psi}{}\big)^\ast=\func{\psi_j}\phi_k=\ketbra{\phi_k}{\psi_j}$. This then entails that 
\begin{equation}
\hsbraket{\E{jk}{\Phi,\Psi}{}}{T}=0,\;\;\forall j,k\in\N\quad\implies\quad T=0,
\end{equation}
i.e., $\HSprod$ is non-degenerate. Moreover, by observing that 
\begin{equation}
\abs{\hsbraket{S}{T}}=\abs{\tr(S^\ast T)}=\Big|\sum_n\braket{\phi_n}{S^\ast T\phi_n}\Big|\leq\No{S}\,\No{T},
\end{equation}
we see that $\HSprod$ satisfies the Cauchy-Schwarz inequality too, i.e., it is a non-Archimedean (or $p$-adic) inner product on $\T(\Hi,\Ki)$ (see Subsection~\ref{subsec.2.1}). We now prove that $\{\E{jk}{\Phi,\Psi}{}\}_{j,k\in\N}$ is an orthonormal basis in $\T(\Hi,\Ki)$. Indeed, it is clear that this set is orthonormal:
\begin{equation}
\hsbraket{\E{mn}{\Phi,\Psi}{}}{\E{jk}{\Phi,\Psi}{}}=\tr\big(\ketbra{\phi_n}{\psi_m}\,\ketbra{\psi_j}{\phi_k}\big)=\delta_{mj}\delta_{nk},
\end{equation}
and, for every finite subset $I\subset\N\times\N$ and $\{\alpha_{ij}\}_{i,j\in I}\subset\Qe$,
\begin{equation}
\Big\|\sum_{i,j\in I}\alpha_{ij}\E{ij}{\Phi,\Psi}{}\Big\|=\max_{i,j\in I}\abs{\alpha_{ij}}.
\end{equation}
Let us now introduce the family of matrix operators $\Top{l}$, $l\in\N$, defined, for every trace class operator  
$T=\op_{\Phi,\Psi}(T_{mn})$, as
\begin{equation}
\Top{l}\coloneqq\sum_{\max\{m,n\}\leq l}T_{mn}\E{mn}{\Phi,\Psi}{}.
\end{equation}
Then, $\forall l\in\N$, $\Top{l}\in\B_{\Phi,\Psi}$, and $\lim_l\Top{l}$ will converge to an operator in $\overline{\B_{\Phi,\Psi}}^{\No{\cdot}}(=\T(\Hi,\Ki))$. We now want to prove that $\Top{l}$ converges precisely to $T$. Indeed, noting that
\begin{equation}
\No{T-\Top{l}}=\Bigg\|\sum_{\max\{m,n\}>l}T_{mn}\E{mn}{\Phi,\Psi}{}\Bigg\|=\sup\big\{\abs{T_{mn}}\mid \max\{m,n\}>l\big\},
\end{equation}
and since $T\in\T(\Hi,\Ki)$ --- which entails that, $\forall\epsilon>0$, there exists $j_\epsilon\in\N$ such that $\max\{m,n\}>j_\epsilon$ $\implies$ $\abs{T_{mn}}<\epsilon$ --- we have, for $l\geq j_\epsilon$
\begin{equation}
\No{T-\Top{l}}=\sup\big\{\abs{T_{mn}}\mid \max\{m,n\}>l\big\}<\sup\big\{\abs{T_{mn}}\mid \max\{m,n\}>j_\epsilon\big\}<\epsilon,
\end{equation}
i.e., $\lim_l\Top{l}=T$. Concluding, $\{\E{jk}{\Phi,\Psi}{}\}_{j,k\in\N}$ is an orthonormal system whose span is dense in $\T(\Hi,\Ki)$, i.e., it is an orthonormal basis in $\T(\Hi,\Ki)$. This concludes the proof as we have proven that the $p$-adic Banach space $\T(\Hi,\Ki)$ can be endowed with a $p$-adic inner product --- i.e., $\HSprod$ --- and the orthonormal basis $\{\E{jk}{\Phi,\Psi}{}\}_{j,k\in\N}$.
\end{proof}

\end{document}